%% file: main.tex
\input{preamble}

\begin{document}

\title{Bandwidth Fee Mechanisms for Certified Transaction Dissemination}

\author{
Eleftheria Fassman\thanks{University of Michigan} 
\and
Mahimna Kelkar\thanks{Purdue University}
\and
Benjamin Marsh\thanks{Sei Labs and University of Portsmouth}
\and
Ke Wu\thanks{University of Michigan}
}

\maketitle
\input{abstract}

\clearpage
\tableofcontents
\clearpage

\input{Sections/all.tex}

\end{document}

%% file: preamble.tex
\input{Conferences/FullVersion/fullversion.tex}
\usepackage{xcolor}
\usepackage{url}
\usepackage{booktabs}
\usepackage{tikz}
\usetikzlibrary{arrows.meta, positioning, shapes.geometric, calc}
\usepackage{amsfonts,amsmath}
\usepackage[capitalize,noabbrev]{cleveref}
\usepackage{mdframed}
\usepackage{aliascnt}
\hypersetup{hypertexnames=false}
\usepackage{amsthm}
\usepackage{amssymb}

\newcommand{\mahimna}[1]{}
\newcommand{\ke}[1]{}
\newcommand{\ellie}[1]{}
\newcommand{\ben}[1]{}
\newcommand{\alex}[1]{}

\newcommand{\mypara}{\paragraph}

\newtheorem{theorem}{Theorem}[section]
\crefname{theorem}{theorem}{theorems}
\Crefname{theorem}{Theorem}{Theorems}

\newaliascnt{definition}{theorem}
\newtheorem{definition}[definition]{Definition}
\aliascntresetthe{definition}
\crefname{definition}{definition}{definitions}
\Crefname{definition}{Definition}{Definitions}

\newaliascnt{lemma}{theorem}
\newtheorem{lemma}[lemma]{Lemma}
\aliascntresetthe{lemma}
\crefname{lemma}{lemma}{lemmas}
\Crefname{lemma}{Lemma}{Lemmas}

\newaliascnt{claim}{theorem}
\newtheorem{claim}[claim]{Claim}
\aliascntresetthe{claim}
\crefname{claim}{claim}{claims}
\Crefname{claim}{Claim}{Claims}

\newaliascnt{assumption}{theorem}
\newtheorem{assumption}[assumption]{Assumption}
\aliascntresetthe{assumption}
\crefname{assumption}{assumption}{assumptions}
\Crefname{assumption}{Assumption}{Assumptions}

\newaliascnt{corollary}{theorem}
\newtheorem{corollary}[corollary]{Corollary}
\aliascntresetthe{corollary}
\crefname{corollary}{corollary}{Corollaries}
\Crefname{corollary}{Corollary}{Corollaries}

\newaliascnt{remark}{theorem}
\newtheorem{remark}[remark]{Remark}
\aliascntresetthe{remark}
\crefname{remark}{remark}{remarks}
\Crefname{remark}{Remark}{Remarks}

\newcommand{\C}{\ensuremath{\mathcal{C}}}
\newcommand{\assign}{\ensuremath{{\sf assign}}}
\newcommand{\price}{\ensuremath{{\sf price}}}
\newcommand{\reward}{\ensuremath{{\sf reward}}}
\newcommand{\band}{\ensuremath{{\sf cap}}}
\newcommand{\block}{\ensuremath{{\bf B}}}

\newcommand{\tx}{\ensuremath{{\sf tx}}}
\newcommand{\util}{\ensuremath{{\sf util}}}

\newcommand{\rev}{\ensuremath{\tau}}

\newcommand{\faults}{t}

\usepackage{algorithm,algpseudocode}
\usepackage{enumitem}

%% file: Conferences/FullVersion/fullversion.tex
\documentclass[11pt]{article}
\usepackage[utf8]{inputenc}
\usepackage[letterpaper,margin=1in,bmargin=1in]{geometry}

\usepackage[pdfauthor={},pdfpagelabels=true,linktocpage=true,hidelinks,linktoc=all,colorlinks]{hyperref}
\hypersetup{linkcolor = {teal}, citecolor = {magenta}, urlcolor = {black}}

\usepackage[style=alphabetic,backend=bibtex,maxbibnames=4,maxnames=4,minnames=3]{biblatex}

\date{}

%% file: abstract.tex
\begin{abstract}
We study bandwidth fee mechanisms (BFMs) for pricing threshold certified transaction dissemination before consensus. BFMs are analogous to transaction fee mechanisms (TFMs) introduced by Roughgarden (2021), except that they price blockchain communication rather than computation. Separately pricing this bandwidth is especially relevant when transaction dissemination is separated from consensus and execution. We model BFMs as a two-sided procurement auction between users who submit transactions and validators who receive and forward them. Any validator may be assigned as the source of a transaction and forward it to a threshold of attesters. The model accounts for the source validator's receipt of the user submitted transaction and the subsequent transfers between validators. It abstracts away unrestricted multi-hop gossip and network topology.

In this work, we focus on settings where transactions must be certified/attested by a threshold of network nodes before consensus, and we develop the following results.
\begin{itemize}[leftmargin = 5mm]
    \item We characterize the limits of myopic incentive compatibility for BFMs: No non-trivial transcript-based BFM can prevent profitable collusion between two validators, and no non-trivial BFM can simultaneously satisfy incentive compatibility for individual user, individual validator, and coalition of one validator and one user.
    \item In contrast, once we restrict attention to unilateral deviations, we present a greedy-route posted-price mechanism that is incentive compatible for every individual user and validator.
\item We study how the bandwidth price should evolve across blocks. Assuming validators execute honestly, we propose an EIP-1559-style price-update rule and show that its fluid approximation has a unique fixed point, which is locally stable for small enough step sizes. In large markets, prices that start near this point stay near it with high probability over any fixed time horizon.
\end{itemize}
\end{abstract}

%% file: Sections/all.tex
\input{Sections/introduction}

\input{Sections/overview}

\input{Sections/simple-model}
\input{Sections/impossibility}

\input{Sections/posted-price}

\input{Sections/ILP}
\input{Sections/multi-blocks}

\input{simulation}

\mypara{Acknowledgments.}
We thank Tim Roughgarden for useful discussions during early stages of this work.
E.~Fassman and K.~Wu were supported in part by a Stellar Academic Research Grant. 

\printbibliography

\appendix
\newpage
\input{Sections/mpc-tfm}

%% file: Sections/introduction.tex
\section{Introduction}
\label{sec:introduction}

Designing modern blockchains requires the use of techniques not only from distributed systems, but also from mechanism design. A prominent example of this is transaction fee mechanism (TFM) design---introduced in seminal work by Roughgarden~\cite{roughgarden2021_tfm, roughgarden2024_tfm}---which is used to price the scarce resource of blockchain \textit{computation}. In the past few years, there has been a rich body of work on the design of TFMs for a variety of settings and assumptions~\cite{chung2024_tfm, chung2023_tfm, wu2024_tfm, shi2023_tfm, bahrani2024_tfm, garimidi2025_tfm, ganesh2024revisiting}. 

In addition to computation, transactions also consume \emph{communication} resources.
Blockchain protocols often involve a phase of transaction \textit{gossip}, that is, the propagation of transaction data to nodes in the network so that the transaction can be included in a subsequent block proposal within blockchain consensus. This set of gossiped but pending transactions is commonly referred to as the \textit{mempool}. Today, mempool transactions are not charged unless they become part of a block proposal and are finalized and executed as part of the blockchain consensus. In other words, transactions are not properly charged for their \textit{bandwidth} usage prior to their execution. Effectively, if a transaction is not included in a block, or is included but its execution fails, it may end up paying significantly less than its resource strain on the protocol. For this paper, ``gossip'' means a threshold-style certified dissemination before consensus.

\mypara{Insufficiency of TFMs.}
In our view, the current model of execution-based pricing (rooted in TFMs) may now be insufficient with evolution in consensus protocol designs. We provide three broad protocol classes as illustrative examples where bandwidth usage of a transaction does not correspond to its execution.

First, a growing line of performance-focused modern protocols decouple transaction dissemination from ordering and consensus---notably, this is seen in popular DAG-based protocols and variants that copy their split layer design~\cite{keidar2021dagrider, danezis2022narwhal, spiegelman2022bullshark, babel2025mysticeti, giridharan2024autobahn}. In these protocols, when a transaction is being propagated through the network, it is unknown whether or not it will execute as intended. Yet, the cost of this propagation is borne by protocol nodes. 

Second, some protocols utilize \textit{multiple} leaders for creating blocks with the objective of providing better robustness-related guarantees such as censorship resistance and fairness~\cite{garimidi2025mcp,kelkar2020order}. Commonly, multiple leaders need to include the same transaction in their own proposal (parameterized by e.g., the level of robustness desired by the user). Here, the cost of transaction propagation is borne multiple times even though it only executes once in the final ordering.

Finally, similar concerns exist in recent protocols that support encrypted mempools~\cite{choudhuri2024bte}. These protocols hide transaction data before execution to prevent leakage-based adversarial behavior~\cite{daian2020flashboys,qin2022bev}---a side effect however, is that a transaction's final execution is unknown during its propagation. 

Motivated by this challenge, we initiate the theory of mechanisms for pricing \emph{communication} in transaction propagation \emph{prior to consensus}.
We focus on settings where transactions must be made available before consensus determines final inclusion, and where such availability can be certified by attestations from network nodes.

\subsection{Our Contributions}
\label{subsec:contributions}
\paragraph{Conceptual Contribution: Formal Modeling of BFMs with a Planner.}
We introduce the \emph{bandwidth fee mechanism} (BFM), a new primitive for pricing communication in decoupled-dissemination blockchain architectures. 
Consider $m$ protocol nodes called the \emph{validators}.
Before a transaction can be included in the finalized ledger, it must first obtain a proof-of-availability (PoA) certificate that contains at least a certain number of validator signatures. 
Specifically, a transaction is routed to a \emph{source} validator, who then forwards it to $k$ \emph{attester} validators for signatures and then assembles a PoA certificate for the transaction. 
This threshold-PoA abstraction captures the pricing-relevant part of many pre-consensus gossip and availability protocols.
The certificate identifies the source and attesters appearing in the public transcript, but does not by itself certify the physical path through which transaction data was transmitted.

Since bandwidth is not directly observable by the protocol, we rely on validators' self-reports. Each validator \(i\) has a private bandwidth capacity \(\mathsf{cap}_i^*\) and a private cost function $f_i^*(x) = c_i^* \cdot g_i(x)$,
where \(g_i\) is a publicly known shape function and \(c_i^* > 0\) is a private scalar. Here, \(g_i(x)\) captures the validator's public technology or resource-consumption profile, whereas \(c_i^*\) captures private cost intensity such as electricity price, opportunity cost, or operating cost. 
For example, if a validator's hardware consumes a roughly fixed amount of energy per transmitted byte, then $g_i(x) = x$ and $c^*_i$ is its private per-unit cost. More generally, $g_i$ can be an externally specified or benchmarked curve for a validator's hardware or network class.
Users submit transactions with bids indicating how much they are willing to pay for finalization.

We study BFMs in a \emph{planner model}. The planner receives users' transactions and the \(m\) validators' reports of bandwidth capacities and private costs, and assigns each accepted transaction to a \emph{route} \((s;A)\), where \(s\) is the source validator and \(A\) is a set of \(k\) attesters. 
For each accepted transaction \(tx\), the planner authenticates the intended route by signing an assignment \((tx,s;A)\) and makes it publicly auditable.
The source receives the transaction through the planner, and forwards it to $k$ attesters to collects their signatures.
This model is motivated by existing blockchain infrastructure: in practice, users commonly access blockchain networks through RPC endpoints or gateway services, which can naturally play the role of routing planners. 
Importantly, the planner neither itself takes part in validator-to-validator transaction dissemination nor forces validators to follow its instructions. It only specifies the intended routing; validators may still misreport their costs or capacities or deviate during dissemination. We do not count bandwidth usage of the planner. 

Abstractly, the planner model isolates the pricing and incentive problem created by communication scarcity while remaining close to existing gateway-based infrastructure; we abstract away arbitrary multi-hop paths, network topology, path specific latency, and link congestion.
In fact, our impossibility results hold even within this planner model.

A BFM specifies: 1.) for each transaction, whether it is disseminated and by which route, 2.) how much each finalized transaction pays, and 3.) how much each validator is rewarded. 
We define three desired incentive-compatibility notions: 
\begin{itemize}[leftmargin = 5mm, itemsep = 0.6pt, topsep = 1.2pt]
    \item \emph{User incentive compatibility} (UIC) requires that no individual user can profit from misreporting or injecting fake transactions;
    \item \emph{$c$-validator incentive compatibility} ($c$-VIC) requires that no coalition of up to $c$ validators can jointly profit by deviating;
    \item \emph{$(c,d)$-side-contract proofness} ($(c,d)$-SCP) requires that no coalition of up to $c$ validators and $d$ users can jointly profit by coordinating. 
\end{itemize}
We formulate these incentive guarantees \emph{ex post}: they must hold for every realized profile of user values, validator costs, and bandwidth capacities, rather than only in expectation over a prior distribution. 
It captures the stronger guarantee that, for any realized types, truthful behavior is optimal without relying on distributional assumptions.

\paragraph{Impossibility against coalitions.}
We next characterize two limitations of incentive compatibility of BFMs against coalitions.
The first concerns validator coalitions in \emph{transcript-based} BFMs: mechanisms whose payments and rewards depend only on the planner assignment and the finalized PoA certificates available on-chain. This abstraction reflects the standard permissionless blockchain model: the protocol can only base rewards on information that is publicly verifiable by all participants, such as bids, assignments, and PoA certificates. Information such as which machine physically sent a packet or which node stored the data are forgeable, and not part of the ledger state. 
Thus, absent additional trusted hardware or verifiable data-possession primitives, any reward rule implementable purely by the replicated ledger is naturally transcript-based.

Within such a mechanism, a coalition of two validators can share signing keys so that the source can locally compute the attester's signature without ever transmitting the transaction data and spending bandwidth. This saves the communication cost while maintaining the same validator rewards.
The same argument holds when multiple validators collude with some users since the above strategy does not change users' allocation and payment.
We call a BFM \emph{trivial} if it never routes any transaction.

\begin{theorem}[Impossibility against validator coalition] 
\label{thm:intro-2vic}
No non-trivial transcript-based BFM satisfies $2$-VIC.
No non-trivial transcript-based BFM satisfies $(2,d)$-SCP for any $d\geq 1$.
\end{theorem}

Having ruled out any mechanism resilient against two or more colluding validators, we ask, how about validator-user coalitions with only one validator? 
Can we design a BFM that achieves UIC, $1$-VIC, and $(1,d)$-SCP? 
Unfortunately, the answer is no.
We develop our second impossibility against coalitions involving only one validator and users. Via a reduction to the MPC-assisted TFM impossibility of~\cite{shi2023_tfm}, we show that

\begin{theorem}[Impossibility against validator-user coalition]
\label{thm:intro-3-prop-impossibility}
Suppose all cost functions satisfy $f^*_i(x) > 0$ for any $x > 0$ and are strictly increasing. Then no non-trivial BFM simultaneously satisfies UIC, $1$-VIC, and $(1,1)$-SCP.
\end{theorem}
The reduction isolates a restricted, TFM-like subgame by fixing validator types and assuming honest dissemination. We reach an intermediate conclusion that any BFM that is UIC, $1$-VIC, and $(1,1)$-SCP must have zero validators' net utility (its reward minus its reported cost). 
The full impossibility then exploits the additional BFM strategy space: if a validator is assigned positive amount of work with a positive cost, it may overreport its cost to earn more reward, which breaks $1$-VIC.
One way to circumvent this impossibility is to assign work only to those zero-cost validators --- but this is not practical in real-world contexts. 

\paragraph{Greedy-route posted-price BFM.}
Motivated by the impossibility results, we sacrifice coalition resistance and focus on unilateral incentive compatibility. 
UIC and $1$-VIC form a natural non-cooperative baseline for the myopic BFM game. Together, they imply that the prescribed honest behavior is an ex-post Nash equilibrium: no individual player can improve its utility by deviating unilaterally. 
We design a \emph{greedy-route posted-price mechanism} that achieves UIC and $1$-VIC, making it robust against any single deviating party.

    \begin{mdframed}
        \begin{center}
            Greedy-Route Posted-Price Mechanism
        \end{center}
        The mechanism fixes a base fee $b > 0$ charged to each finalized transaction and a reward parameter $0 < \tau < b$ being the reward scalar of validator's verifiable work. 
        
        Validators with reported cost scalar $c_i \geq \tau$ are excluded. The planner scans all $L = m \cdot \binom{m-1}{k}$ possible routes in a fixed lexicographic order and greedily accepts copies of each route so long as two conditions hold: 
        \begin{itemize}[leftmargin = 5mm, itemsep = 1pt, topsep = 1.2pt]
            \item \textit{Capacity:} No validator's assigned bandwidth cost exceeds its reported capacity.
            \item \textit{Budget:} The incremental reward cost of one more copy of the route does not exceed $b$. 
        \end{itemize}
        Incoming transactions are processed in a priority order that cannot be manipulated by fake-bid injection, such as the timestamp when the planner receives the transaction.
    \end{mdframed}

\begin{theorem}[Incentive compatibility of greedy-route posted-price mechanism]
\label{thm:intro-main-greedy}
For any $m \geq 3$ and $0 < \tau < b$, the greedy-route posted-price mechanism satisfies UIC and $1$-VIC.
\end{theorem}

We also compare the greedy mechanism against the optimal integer-program allocation subject to the base fee $b$ and reward parameter $\tau$. 
In the capacity-limited regime, when the final greedy allocation is capacity-maximal so that every route is blocked by some validator's remaining capacity, the greedy-route BFM is a $(2k+1)$-approximation to the optimal throughput, with an additive slack of $\frac{(k+1)m}{2k+1}$.

\paragraph{Price update and dynamic stability.}
Having established the greedy-route posted-price mechanism in the myopic setting, we next turn to how the base fee should be updated across blocks.
Here we assume that validators report their types once and then execute their assignments honestly, and we focus on the capacity-limited regime, in which the route-level budget constraint is slack near the operating point.
We design an EIP-1559-style update rule for this setting and study its local behavior around the operating point.

Let $\Theta^{\max}$ denote the maximum aggregate bandwidth the greedy-route mechanism can disseminate under unlimited demand.
In each round $t$, let $\Theta_t$ denote the bandwidth available at the beginning of the round and $\Gamma_t$ the bandwidth consumed during it.
Used bandwidth recovers at a constant rate $\alpha$.
The base fee $B_t$ targets a usage of $\Theta_t/2$, half of the currently available bandwidth:
\[
  B_{t+1} =
  \begin{cases}
    \displaystyle
    B_t\left(1 + d\cdot\frac{\Gamma_t - \Theta_t/2}{\Theta_t/2}\right) & \text{if } \Theta_t > 0,\\
    B_t & \text{if } \Theta_t = 0,
  \end{cases}
\]
where $d \in (0,1)$ is the step size.
\cref{sec:overview-dynamic} states the same rule in terms of the normalized quantities $X_t := \Theta_t/\Theta^{\max}$ and $G_t := \Gamma_t/\Theta^{\max}$.

\begin{theorem}[Informal: local stability and finite-horizon tracking]
Suppose that in each round $\mathrm{Poisson}(\lambda\Theta^{\max})$ users arrive, with valuations drawn i.i.d.\ from a continuous, strictly increasing distribution $F$ that is differentiable with positive density near the operating price, and that the arrival rate $\lambda$ is large enough to sustain a positive price. Then:
\begin{enumerate}
  \item The fluid approximation of the dynamic BFM has a unique positive-price fixed point $(x^*, b^*)$, at which bandwidth usage equals the target and recovery exactly replenishes consumption.
  \item This fixed point is locally asymptotically stable for every step size $d < \min\{1, d_{\mathrm{crit}}\}$, where $d_{\mathrm{crit}}$ depends explicitly on the recovery rate $\alpha$ and the local demand elasticity at $b^*$.
  \item For every fixed horizon $T$ and every sufficiently small neighborhood $U$ of $(x^*, b^*)$, there is a neighborhood $V \subseteq U$ such that the stochastic process started in $V$ stays in $U$ through round $T$ with probability tending to one as $\Theta^{\max} \to \infty$. Over this horizon, its deviation from the fluid trajectory is $O_p(1/\sqrt{\Theta^{\max}})$.
\end{enumerate}
\end{theorem}
We develop this result via a fluid approximation in \cref{sec:overview-dynamic}. The guarantee is local and holds over fixed finite horizons. 

We further complement the analysis with numerical simulations in the full version. The experiments show that 1.) the stochastic process closely track the fluid fixed-point predictions, 2.) the normalized deviations shrink as market size grows, and 3.) the predicted local stability threshold is visible in simulated price trajectories.

\subsection{Additional Related Work and Future Directions}
\label{sec:related}
\smallskip\noindent{\bf Transaction Fee Mechanisms (TFMs).}
As discussed previously, TFMs study how to price blockchain computation among competing transactions. Beginning with the foundational treatments of Roughgarden~\cite{roughgarden2021_tfm,roughgarden2024_tfm} and Chung and Shi~\cite{chung2023_tfm}, a substantial literature has developed around incentive-compatible TFMs~\cite{wu2024_tfm,shi2023_tfm,chung2024_tfm,bahrani2024_tfm,garimidi2025_tfm, ganesh2025characterizing,ganesh2024revisiting, ferreira2024incentive}. The central design goals are user incentive compatibility (UIC), miner incentive compatibility (MIC), and various notions of collusion resistance such as side-contract proofness (SCP) and off-chain agreement proofness (OCA-proofness). A recurring theme is a fundamental tension between UIC and collusion resistance---e.g., ~\cite{chung2023_tfm} showed that UIC and 1-SCP are incompatible in the plain model, ~\cite{shi2023_tfm} proved UIC and 2-SCP are incompatible even with cryptography in the MPC-assisted model, and~\cite{chung2024_tfm} showed it is impossible to achieve UIC, MIC, and OCA-proofness together. 

Although BFMs borrow several incentive notions from the TFM literature, the underlying resource-allocation problem are different. In a TFM, users compete for a \emph{publicly known, fixed scalar} amount of block space, whereas in a BFM, the scarce resource being the validators' private bandwidth, and can only be accessed by the mechanism through validators' self report. 
This fundamentally different scarce resource also changes the strategy space. In MPC-assisted TFM, a miner can only inject fake bids. Once bids are submitted, the allocation and payment rules are faithfully executed. 
In a BFM, validators may strategically deviate in how they report their capacity and cost that influence the scarce resource and deviate in executing the assigned dissemination.

\smallskip\noindent{\bf Propagation incentives and fee dynamics.}
Babaioff et al.~\cite{babaioff2012balloons} study incentives to forward Bitcoin transactions and design propagation rewards in a peer to peer network. Our mechanism prices the bandwidth of an assigned source and threshold attesters, using the resulting PoA certificate for accounting. Leonardos et al.~\cite{leonardos2021dynamical} give conditions for convergence of EIP-1559's base fee and show that larger update steps can produce chaos. Their subsequent work~\cite{leonardos2024optimality} bounds average block utilization near its target even without fee convergence. Our usage target changes with available validator bandwidth, our dynamic results establish local fluid stability and finite horizon stochastic tracking.

\smallskip\noindent{\bf Two-sided and multidimensional procurement auctions.}
Our model is also related to classical procurement auctions~\cite{laffont1986cost,mcafee1986bidding}, where a buyer procures a service from agents with private costs. 
A large literature studies richer procurement settings in which suppliers differ along multiple dimensions, such as price, quality, or correlated cost components. 
Canonical examples include multidimensional procurement auctions and scoring auctions~\cite{che1993design,branco1997design,asker2008properties}, where bids specify both monetary and non-monetary attributes and the buyer evaluates offers through a scoring rule. 
More broadly, multidimensional screening is known to be substantially harder than single-parameter mechanism design~\cite{rochet1987rationalizability, rochet1998ironing}. 

BFMs share this multi-parameter flavor because validators report both bandwidth capacities and cost information. 
However, BFMs differ from classical procurement in several blockchain-specific ways. 
First, the allocation is combinatorial: each transaction must be routed through a source validator and a set of attesters, so the mechanism must reason about network routes and bandwidth constraints rather than selecting a single supplier for each transaction. 
Second, validators are not passive suppliers who simply deliver a product after being selected; they may deviate during the dissemination process (the ``service'') itself. 
Finally, the protocol observes only the public ledger transcript, such as assignments and PoA certificates, rather than the physical packet transfers that generated it. 
Thus, while classical procurement and multidimensional screening provide useful benchmarks, BFMs require new definitions and impossibility arguments tailored to certified communication in blockchain protocols.

\smallskip\noindent{\bf Future directions.}
We position our work as the first to study bandwidth pricing and leave several directions open. 
A natural first one is to remove or decentralize the planner: instead of a single planner globally assigns routes, can route selection be implemented by validators themselves in a decentralized manner or by a consensus-level committee while preserving incentive compatibility? 
Second, our dynamic analysis assumes that validators report once and then execute honestly in later rounds. Understanding farsighted validator or user incentives across multiple rounds is an important next step.
Third, our impossibility results rule out coalition robustness for transcript-based BFMs. An interesting future direction is to identify additional assumptions or primitives, such as non-delegable data receipts or verifiable data-possession challenges, under which meaningful approximate or restricted coalition incentive compatibility becomes possible.

%% file: Sections/overview.tex
\section{Technical Overview}
\label{sec:technical-overview}

This section details our model and highlights the core technical aspects of our main contributions. 
\cref{sec:overview-model} introduces a more detailed model, strategy spaces, utilities, and incentive-compatibility notions. 
\cref{sec:overview-impossibility} explains the techniques used in the two impossibility results and \cref{sec:overview-greedy} presents the greedy-route BFM.
\cref{sec:overview-dynamic} studies multi-round price dynamics.

\subsection{Model Overview}
\label{sec:overview-model}
Let $m$ be the total number of validators and $k$ be the number of attesters needed for each transaction to get a PoA certificate.
For the technical overview, we abstract away the consensus part and assume that transactions with a valid PoA certificate are finalized.
Each validator has a private bandwidth capacity $\mathsf{cap}^*_i$ (the maximum number of bandwidth units they can spend) and a private cost function $f^*_i(x) = c_i \cdot g_i(x)$, where $g_i$ is a \emph{publicly known} shape function capturing how cost scales with bandwidth usage, and $c_i > 0$ is a private scalar. We assume $g_i(0) = 0$ and that $g_i(x)$ is strictly increasing.

In each round, a stream of users arrive, each wishing to have one transaction committed. Each user $u$ has a true value $b^*_u$ of its transaction, which specifies the maximum amount it is willing to pay for finalizing the transaction. Each user submits a bid to compete with other users for PoA. We assume a \emph{centralized planner} who receives users' transactions and validators' reports on bandwidth capacities and private costs and implements a bandwidth fee mechanism (BFM) that specifies:
\begin{itemize}[leftmargin = 5mm, itemsep = 1pt]
    \item An {\bf allocation rule} that given reported validator capacities and costs and user bids, assigns each transaction to a route. A route $(s; A = \{a_1, \ldots, a_k\})$ consists of a source validator $s$, who is responsible for forwarding it to $k$ attester validators $a_1, \ldots, a_k$, each of whom stores the data and returns a signature. 
    Here, $s\notin A$ and $\{s, a_1,...,a_k\} \subseteq [m]\cup\{\bot\}$, where $\bot$ means the transaction is not assigned for dissemination this round.
    The source then assembles the $k$ signatures into a PoA certificate.
    \item A {\bf payment rule} specifying how much to charge each finalized transaction.
    \item A {\bf reward rule} that specifies how much reward each validator gets.
\end{itemize}
A BFM must satisfy the following three feasibility requirements. 
\begin{itemize}[leftmargin = 5mm, itemsep = 1pt]
    \item {\it Attestation feasibility}: 
    Each transaction gets assigned at most one route.
    If $s = \bot$, then all $a_1,...,a_k = \bot$; otherwise if $s\in[m]$, then $A\subset[m]$.
    \item {\it Budget feasibility:} the total validator rewards should not exceed total user payments.
    \item {\it Individual rationality:} each user's payment should not exceed its bid, and each validator $i$'s assigned bandwidth usage $\beta^{\rm ass}_i$ should satisfy that $\beta^{\rm ass}_i \leq \band_i$ and that $r_i \geq f_i(\beta^{\rm ass}_i)$, where $(\band_i,f_i)$ is validator $i$'s reported bandwidth capacity and cost functions, and $r_i$ is the expected reward it receives.
\end{itemize}

Without loss of generality, we normalize the bandwidth usage and assume that each transaction consumes one unit of bandwidth, and that sending signatures costs zero bandwidth. This normalization is without loss of generality for our positive results. 
Our greedy-route posted-price mechanism extends naturally to transactions of different sizes, and our lower bound holds even if every transaction has the same unit size.

\smallskip\noindent\textbf{Running Example.}
Throughout this section, we will use a running example with $m = 3$ validators and a parameter $k = 2$. There are $L = m\binom{m-1}{k} = 3 \cdot \binom{2}{2} = 3$ possible routes. We write a route as $(s;A)$ where $s \in \{1,2,3\}$ is the source validator and $A \subseteq \{1,2,3\} \setminus \{s\}$ is the set of $k = 2$ attesters. The three routes are:
\[
    R^1 = (1; \{2,3\}), \qquad R^2 = (2; \{1,3\}), \qquad R^3 = (3; \{1,2\}).
\]
For example, $R^1 = (1; \{2,3\})$ means validator $1$ is the source and validators $2$ and $3$ are the attesters.

A transaction assigned to a non-$\bot$ route consumes $q = 2k+1 = 5$ total bandwidth units: $k+1 = 3$ units from the source (one unit receiving \textsf{tx} from the planner, plus two units sending \textsf{tx} to the attesters), and one unit from each attester receiving \textsf{tx} from the source. We define the bandwidth weight $w_i(R)$ of validator $i\in[m]$ in route $R = (s; A)$ where $s\notin A$ as the number of bandwidth units $i$ spends if that route is used:
\begin{equation}
    \label{eqn:w-band}
    w_i(R) := (k+1)\mathbf{1}_{\{i = s\}} + \mathbf{1}_{\{i \in A\}}.
\end{equation}

\smallskip\noindent\textbf{Utilities and Strategy Space.}
For the incentive analysis, we focus on \emph{myopic} players who only care about their utility in a single round.
A user $u$ with true value $b^*_u$ receives utility ${\bf 1}_{\{u\text{'s transaction finalized}\}}\cdot b^*_u - p_u,$
where $p_u$ is the total payment the user paid. 
A validator $i$ who receives reward $r_i$ and spends $\beta_i$ units of bandwidth receives utility $r_i - f^*_i(\beta_i)$. A validator that uses more bandwidth than its true capacity $\mathsf{cap}^*_i$ crashes and receives utility $-\infty$.

An honest user submits a single bid equal to its true value. A strategic user may submit any number of bids, including fake bids, or may misreport its value. An honest validator reports its true cost function and capacity and follows the assigned routes for dissemination. A strategic validator may misreport its cost scalar $c_i$ or capacity, inject fake bids as a user, deviate from the assigned dissemination rules, or crash intentionally. A coalition's strategy is a combination of all coalition members' strategies based on their joint view. Under a joint view, members share all private information, including validators' signing keys. 

\smallskip\noindent\textbf{Incentive Compatibility.}
We study the following incentive properties:
\begin{itemize}[leftmargin = 5mm, itemsep = 1pt]
    \item {\bf User incentive compatible (UIC):} No individual user can increase its utility compared to honest case by deviation, assuming validators behave honestly.
    \item {\bf $c$-validator incentive compatible ($c$-VIC):} No coalition of up to $c$ validators can strategically increase their joint utility.
    \item {\bf $(c,d)$-side-contract-proof ($(c,d)$-SCP):} No coalition of up to $c$ validators and at least one and up to $d$ users can increase their joint utility by coordinating. 
\end{itemize}

\subsection{Impossibility against Coalitions}
\label{sec:overview-impossibility}

Ideally, one would hope to design a BFM achieving all three desired incentive-compatibility properties. However, we show the following impossibilities.

\smallskip\noindent\textbf{Impossibility against colluding validators (\cref{thm:intro-2vic}).}
We show that no non-trivial \emph{transcript-based} BFM can be resilient against a coalition of two validators, regardless of how the mechanism is designed. 
A BFM is \emph{transcript-based} if, after reports and bids are fixed, payments and rewards depend only on the planner assignment and the finalized PoA certificates. That is, the mechanism cannot verify which validators actually transmitted or received data, but only which signatures appear in the final certificates. 

Within a transcript-based BFM, a coalition of two validators can share their private signing keys with each other and delegate the signing. 
Now if validator $s$ is assigned to source a transaction and validator $a$ is an assigned attester, 
the source $s$ can form a coalition with the attester and perform the following strategy:
$s$ locally computes $a$'s signature using $a$'s signing key without ever sending the transaction to $a$. Otherwise, they behave honestly.
The PoA certificate that $s$ assembles looks identical to an honestly produced one, so the mechanism cannot detect the deviation. However, the coalition has saved the bandwidth cost of sending the transaction from $s$ and receiving by $a$. So, as long as the validator incurs a positive cost of transmitting one unit of data, this strategy of delegating signatures is profitable.
The same result holds for any validator-user coalition with two or more validators.

\smallskip\noindent\textbf{Impossibility against a single validator colluding with users (\Cref{thm:intro-3-prop-impossibility}).}
Ruling out coalitions with multiple validators, we will now focus on the case where there is only one validator in the coalition.
We show a separate impossibility that no non-trivial BFM can simultaneously achieve UIC, $1$-VIC, and $(1,1)$-SCP, where $(1,1)$-SCP means that any coalition of one validator and one user cannot jointly profit by deviating.

The proof reduces the problem to an impossibility result for \emph{Multi-Party Computation (MPC)-assisted-TFMs} due to Shi, Chung, and Wu~\cite{shi2023_tfm}. An MPC-assisted TFM works as follows. There is a fixed block-space limit $\Theta$. An ideal functionality $\mathcal{F}_{\mathsf{MPC}}$ (which always honestly executes as specified) receives the full bid vector from all players (possibly injected fake bids as well), selects at most $\Theta$ transactions to confirm according to the allocation rule, charges each confirmed user according to the payment rule, and transfers a total validator revenue to the validators according to the revenue rule. The key feature is that validators cannot override the allocation or payment rules. Once bids are submitted, the outcome is determined entirely by $\mathcal{F}_{\mathsf{MPC}}$. The only remaining strategic validator behavior is to submit a strategic bid vector, including fake bids.

Now we explain the reduction.
Fix the validators' true types $(\band^*_i,f^*_i)$ for each validator $i$. Suppose all validators report honestly and execute the assigned dissemination honestly. Under these conditions, the BFM $M$ induces an MPC-assisted TFM $\overline{M}$ with block size $\Theta = \sum_i \mathsf{cap}^*_i$: The functionality $\mathcal{F}_{\mathsf{MPC}}$ receives bid vector, runs the allocation rule of BFM. Any transaction with a non-$\bot$ route is confirmed, and each transaction's payment is the same as in the BFM. The revenue each validator $i$ gets is $r_i - f_i^*(\beta^{\rm ass}_i)$, where $r_i$ is the reward of validator $i$ in the reward rule of BFM, and $\beta^{\rm ass}_i$ is the assigned bandwidth usage.
Then the BFM and the induced TFM produce the same finalized transactions, payments, and players' utilities.

Since the only remaining strategy for strategic players is the ability to submit strategic bids in both the BFM and the MPC-assisted TFM, any profitable deviation for a coalition or strategic individual $\C$ in $\overline{M}$ via fake bids translates directly into a profitable deviation of $\C$ in $M$.
Therefore, if a BFM achieves UIC, $1$-VIC, and $(1,1)$-SCP, then the induced TFM also achieves UIC, $1$-VIC (called MIC in the original work), and $(1,1)$-SCP.
However, the impossibility result in \cite{shi2023_tfm} implies that such TFM must always have zero validator utility on any bid vector.

Zero validator utility in $\overline{\mathcal{M}}$ means that in the BFM $\mathcal{M}$, each validator $i$'s utility must be zero in an honest execution on any input bid vector. 
That is, $r_i = c^*_i\cdot g_i(\beta^{\mathrm{ass}}_i)$ on every bid vector. However, since the cost function is self-reported, this creates a profitable deviation. By non-triviality, there exists a validator $i$ that is assigned a positive amount of work in an honest execution ($\beta^{\mathrm{ass}}_i > 0$) with positive probability. 
Now consider a world where validator $i$'s true cost parameter is some $c_i$, where $0 < c_i < c^*_i$. 
In the honest case, validator $i$'s utility is zero. However, in the strategic case, if validator $i$ overreports its cost to $c^*_i$, it gets assigned $\beta^{\mathrm{ass}}_i$ amount of work and receives reward $c^*_i \cdot g_i(\beta^{\mathrm{ass}}_i)$ with a positive probability.
So validator $i$ gets strictly positive expected utility by over-reporting.
This contradicts 1-VIC. So no non-trivial BFM simultaneously achieves UIC, 1-VIC, and $(1,1)$-SCP when all cost functions satisfy $f^*_i(x) > 0$ for $x > 0$ and are strictly increasing.

\subsection{Greedy-Route BFM}
\label{sec:overview-greedy}

Despite the above impossibility results that rule out incentive compatibilities against coalitions, we design a greedy-route posted-price mechanism that achieves UIC and 1-VIC, so no single user or validator can
profitably deviate.

\smallskip\noindent\textbf{The Posted-Price Framework.}
Our mechanism follows a posted-price framework.
It fixes two protocol parameters: a base fee $b > 0$ and a reward parameter $\tau < b$. 
A user transaction is \emph{eligible} if and only if its bid is at least $b$. 
Every finalized transaction pays exactly $b$ to the mechanism, regardless of the user's bid. 
Validators with reports $c_i \geq \tau$ are assigned no work since their per-unit cost exceeds the per-unit reward.

Every validator $i$ that completes its full assigned \emph{source work} receives a reward of $\tau \cdot g_i(\beta_i)$, where $\beta_i = (k+1)B_i + A_i$ with $B_i$ being the number of finalized transactions sourced from $i$ and $A_i$ is the number of attestations validator $i$ provides for other transactions in the finalized ledger.
If a validator fails to complete its assigned source work, it receives nothing.
Note that we do not need to check whether a validator finished its assigned attestation because receiving that transaction already consumes its bandwidth.
If it refuses to sign that transaction, it spends one unit of bandwidth without receiving any rewards for that unit.
Therefore, a validator is always incentivized to finish its assigned attestation work.

\smallskip\noindent\textbf{A Strawman Attempt: Integer Program (IP)-Based Allocation.}
Given the above posted-price framework, a natural thought of designing the allocation rule would be to solve an integer program (or its relaxation) that maximizes the number of finalized transactions subject to each validator's capacity constraint and a budget-feasibility constraint. 

Concretely, let $y_x \geq 0$ denote the number of transactions routed along route $R^x$ for $x \in [L]$, and let $\beta^{\rm ass}_i = \sum_{x \in [L]} y_x \cdot w_i(R^x)$ denote the resulting assigned bandwidth for validator $i$, where $w_i(\cdot)$ is defined in \eqref{eqn:w-band}. 
Let $n$ denote the number of eligible transactions.
The integer program solves non-negative integer solutions $y_x$ for $x\in[L]$ to maximize $Y = \sum_{x \in [L]} y_x$ subject to the following constraints:
\begin{equation*}
    \begin{aligned}
    & \beta^{\rm ass}_i \leq \mathsf{cap}_i,\  \forall i \in [m], &\text{(each validator's assigned work no more than capacity)} \\
    & \tau \cdot \textstyle\sum_{i \in [m]} g_i(\beta^{\rm ass}_i) \leq b \cdot Y, &\text{(total rewards no more than total payment)}
    \end{aligned}
\end{equation*}
If there are multiple optimal solutions, we just pick the first solution subject to lexicographic order.
This finds the globally optimal assignment of transactions to routes given the reported types. However, this approach is vulnerable to fake-bid injection by a strategic validator, as we will explain below.

Consider the example with three possible routes $R^1=(1;\{2,3\})$, $R^2=(2;\{1,3\})$, $R^3=(3;\{1,2\})$.
Let the capacities for three validators be $\band^*_1 = 4, \band^*_2 = 4$, and $\band^*_3 = 2$.
The cost shape $g_1(x) = 1.2x^2$, $g_2(x)= g_3(x) = 0.1 x$.
The private cost scalar is $c_1^* = 0.1$, and $c_2^* = c_3^* = 0.9$.
The base fee $b = 10$ and reward $\tau = 1$.

Suppose in the honest case, there is only $n = 1$ eligible transaction. 
The IP assigns the transaction to $R^2$ because $R^1$ is budget-infeasible and $R^3$ is capacity-infeasible.
So validator $1$ gets utility $(\tau-c^*_1)\cdot g_1(1) = 1.08$.
Now imagine that validator $1$ injects a fake eligible bid and make $n = 2$, and otherwise follows the protocol honestly.
The IP can now select both $R^1$ and $R^2$, each for one transaction.
This satisfies both the capacity and budget constraints.
Now validator $1$'s utility is $(\tau-c^*_1)\cdot g_1(4) - b = 7.28$, which is strictly higher than the honest utility of validator $1$.

In general, if a validator $i$ injects fake eligible bids and inflates demand,
the IP may reoptimize globally over all routes and can shift the entire allocation in response. As a result, it may assign validator $i$ more routes than before.
Even if the validator pays the base fee $b$ for the fake transaction, it may still profit from the extra reward due to the global optimization.

\paragraph{Our Greedy-Route Mechanism.}
The above failed attempt implies that fake demand should not cause a global reoptimization that substantially increases the deviating validator's reward.
This motivates our greedy-route mechanism.

Order all $L = m\cdot \binom{m-1}{k}$ routes lexicographically as $R^1, R^2, \ldots, R^L$ (in our example $L = 3$). The planner scans this list in order, and keeps track of the current assigned bandwidth usage $\beta_i$ for validator $i$.
For each route $R^x = (s; A)$, the planner repeatedly accepts copies of $R^x$ (each copy handling one transaction) as long as two conditions hold:
\begin{itemize}[leftmargin = 5mm, itemsep = 1pt]
    \item {\bf Capacity:} Adding one more copy of $R^x$ does not push any validator's total assigned bandwidth above its reported capacity: $\beta_i + w_i(R^x) \leq \mathsf{cap}_i$ for all $i$.
    \item {\bf Budget feasibility:} The incremental reward cost of one more copy of $R^x$ does not exceed $b$: $\tau \cdot \sum_i (g_i(\beta_i + w_i(R^x)) - g_i(\beta_i)) \leq b$.
\end{itemize}
The algorithm stops either when the number of accepted routes equals the number of eligible transactions $n$, or when the algorithm scans through all routes. Incoming eligible transactions are processed in a fixed \emph{priority order} that cannot be manipulated by injection, such as the timestamp when the planner receives a transaction. When the transactions are processed in order of receipt time, the greedy-route mechanism can be used in a streaming manner.
Although the route list has $L = m\binom{m-1}{k}$ entries, the scan never needs to enumerate it. A dynamic program over candidate attesters finds each next feasible route in $O(m^2k)$ time, so the allocation runs in polynomial time and can still be computed online. We present this algorithm in the full version.

The greedy mechanism avoids the vulnerability shown in the IP example through a key structural property: it accepts route $R$ only when:
\[
\tau\sum_{i}\bigl(g_i(\beta_i+w_i(R))-g_i(\beta_i)\bigr) \le b.
\]
This bounds the total new reward generated across all validators by one more route to at most $b$. 
Thus, by injecting a fake transaction, one validator can get at most $b$ reward in total, but it then needs to pay for the base fee $b$, which makes injection non-profitable.
In our earlier example where validator $1$ profits from injecting a fake transaction, now with the greedy-route mechanism, $R^2$ will be chosen for the first transaction, and the algorithm stops. 
Regardless of whether this transaction is validator $1$'s fake transaction, validator $1$ cannot increase its profit. The formal proof that this mechanism satisfies UIC and 1-VIC for any $m \geq 3$ and $0 < \tau < b$ is given as Theorem~\ref{thm:intro-main-greedy}.

\paragraph{Incentive Properties of Greedy-Route Posted-Price Mechanism}
We now sketch the proof of \cref{thm:intro-main-greedy}.

\smallskip\noindent{\bf UIC}. 
Fix a user $u$ with true valuation $b^*_u$. 
We first show that injecting fake bids does not help.
The planner assigns a fixed priority to incoming bids before running the greedy algorithm. Under the priority rule (e.g., planner-receipt time or random permutation with a fixed random seed), injecting additional fake bids cannot improve the priority position of $u$'s real transaction. 
Moreover, the user may need to pay for its fake transactions that gets committed. Therefore, injecting fake bids is not helpful.
Then we show that misreporting does not help. 
Every finalized transaction pays the fixed price $b$ regardless of its stated bid. The bid affects only eligibility.
A user with $b^*_u \ge b$ is finalized whenever the greedy reaches its slot, and overbidding cannot displace a higher-priority transaction. A user with $b^*_u < b$ cannot gain by bidding $\geq b$, since it would pay $b > b^*_u$ and receive negative utility.
In both cases deviating is not profitable.

\smallskip\noindent{\bf 1-VIC}. A validator's strategy is a combination of 1.) deviating in dissemination; 2.) fake bid injection; and 3.) misreporting cost; 4.) misreporting bandwidth.
In the proof, we rule out these components sequentially.
It should be understood as eliminating components of an arbitrary joint deviation, rather than treating each component as isolated deviations. Since crashing results in $-\infty$ utility, we will only consider strategies without crashing.

\smallskip\noindent{\it 1.) Deviating in dissemination is not profitable.}
A validator only gets paid if it finishes all the assigned source work, and zero otherwise.
Therefore, a validator is always incentivized to finish its source work, i.e., sending its assigned transaction to the corresponding attesters for signatures.
Meanwhile, its reward strictly increases based on the number of attestations it provides, but sending a signature costs no bandwidth.
Therefore, a validator is not incentivized to refuse to attest for other validators.
Finally, any extra spam messages will only consume its bandwidth without creating any reward. 
Therefore, deviating in dissemination is not profitable, and we only need to consider strategies where validators behave honestly during dissemination.

\smallskip\noindent{\it 2.) Fake-bid injection is not profitable.}
The planner sorts incoming bids by a fixed priority rule before running the greedy algorithm. The rule is chosen so that injecting fake bids cannot improve the priority position of any real transaction. Since the greedy mechanism's route order and acceptance conditions depend only on the bandwidth vector ($\beta_i$ for each validator $i$), the set of accepted routes is the same prefix regardless of how many fake bids are present.

Suppose in the honest case the algorithm accepts $|\mathcal{R}^*|$ routes. After injection it accepts at most $|\mathcal{R}^*|+\Delta$ routes for some $\Delta \leq n'$, since each fake bid triggers at most one additional route. By the budget condition, the total incremental reward of each new route is at most $b$, and since every validator on a route earns strictly positive incremental reward, validator $i$'s share is strictly less than $b$ per route. Since each of the $\Delta$ fake bid costs $b$, the net gain is strictly negative.

\smallskip\noindent{\it 3.) Cost misreporting is not profitable.}
The reward parameter $\tau$ creates an eligibility threshold as the posted-price's role to users. 
A validator with $c^*_i < \tau$ earns a positive per-unit surplus $\tau - c^*_i > 0$ on every unit of assigned work, so it wants to be included and should report honestly. If it instead reports $c_i \geq \tau$, it is excluded and earns zero, and a validator with $c^*_i \geq \tau$ earns a non-positive surplus on every assigned unit; that is, reporting $c_i < \tau$ to gain inclusion can only decrease its utility. We show in both cases misreporting the cost scalar does not improve the validator's utility.

\smallskip\noindent{\it 4.) Bandwidth misreporting is not profitable.}
A structural property of the greedy algorithm is its monotonicity. If validator $i$ reports a higher capacity, its assigned bandwidth $\beta^{\rm ass}_i(\mathsf{cap}_i)$ cannot decrease. Underreporting $\mathsf{cap}_i < \mathsf{cap}^*_i$ therefore weakly reduces $\beta_i$, and since the per-unit surplus $\tau - c^*_i > 0$ is positive for included validators, this strictly reduces their earned reward. Overreporting $\mathsf{cap}_i > \mathsf{cap}^*_i$ either causes the validator to be assigned more work than it can complete, in which case it crashes and earns $-\infty$ utility, or the extra reported capacity is never used by the greedy algorithm, in which case the allocation is identical to truthful reporting. In either case, overreporting does not help.

\paragraph{Approximation Guarantee}
The greedy mechanism is not optimal in general, since it commits to a fixed route order. We compare it against the best possible allocation of the IP for a fixed base fee $b$ and reward parameter $\tau$.
Let $\mathrm{OPT}_\tau$  be the optimal value of $Y$ returned by the IP and $G_{\tau}$ be the number of committed transactions in greedy-route mechanism when there are infinitely many eligible bids.
Now, if every route in the optimal solution is blocked by a validator at capacity, then
\[
    G_\tau \geq \frac{1}{2k+1}\mathrm{OPT}_\tau - \frac{(k+1)m}{2k+1}.
\]
This holds because each validator $i$ can block at most $\beta^G_i + k + 1$ optimal transactions, where $\beta^G_i$ is its final assigned workload under the greedy mechanism.
Summing over validators and using the fact that every greedy transaction consumes exactly $q = 2k+1$ bandwidth units gives the bound above. The additive slack $\frac{(k+1)m}{2k+1}$ arises because each validator can have at most $k+1$ units of capacity unused at the boundary when it first becomes the blocking constraint. 

We do not model the approximation in the
cost-congested regime because a fixed route order is not enough to guarantee a meaningful approximation independent of the cost profiles. 
An unlucky ordering can allocate early transactions to routes whose marginal rewards are high, therefore making later routes fail the marginal budget check, even though a different ordering of the routes might accept more transactions if some lower-cost validators appeared first.
Obtaining an approximation guarantee in this regime requires cost-aware route ordering so as to incentivize truthful validator reporting; we leave this to future work.

\subsection{Dynamic Price Update}
\label{sec:overview-dynamic}

\paragraph{The Multi-Round Model.}
We next design a multi-round base-fee update rule and analyze its local behavior around an operating point.
Validators report their types once at $t = 0$, the reward parameter \(\tau\) remains fixed, and validators execute the greedy-route assignments honestly in all subsequent rounds. 
For the dynamic analysis, we focus on the capacity-limited regime in which the route-level budget constraint is slack in a neighborhood of the operating point.
Thus, after validator reports and the reward parameter \(\tau\) are fixed, the mechanism can be summarized by a fixed maximum aggregate bandwidth
\(\Theta^{\max}\), and the base fee affects demand but not the set of feasible routes near the operating region.

Let \(\Theta_t\in[0,\Theta^{\max}]\) be the available bandwidth at the beginning of round \(t\), and define the normalized supply \(X_t:=\Theta_t/\Theta^{\max}\). If normalized bandwidth consumption in round \(t\) is \(G_t\in[0,1]\), then supply evolves as
\[
X_{t+1}=X_t-G_t+\alpha(1-X_t),
\]
where \(\alpha\in(0,1]\) is the recovery rate. Thus used bandwidth is subtracted immediately, while a fraction \(\alpha\) of the gap to maximum capacity recovers before the next round.

\smallskip\noindent {\bf Base-Fee Update and Stochastic Demand.}
The base fee \(B_t\) is updated similar to EIP-1559, except that the target changes every round rather than a fixed block size:
\[
  B_{t+1} =
  \begin{cases}
    \displaystyle
    B_t\!\left(1 + d\cdot\frac{G_t - X_t/2}{X_t/2}\right) & \text{if } X_t > 0,\\[6pt]
    B_t & \text{if } X_t = 0,
  \end{cases}
\]
where $d\in(0,1)$ is the step size. Hence the price increases when consumption exceeds $X_t/2$ and decreases when consumption is below $X_t/2$.

In each round, we assume $N_t\sim \mathrm{Poisson}(\lambda\Theta^{\max})$
users arrive independently, and each user has valuation drawn independently from a distribution with CDF $F$. A user is eligible at price $B_t$ if its valuation is at least $B_t$. By Poisson thinning, the number of eligible users satisfies $S_t\sim \mathrm{Poisson}\bigl(\lambda\Theta^{\max}(1-F(B_t))\bigr).$
The normalized bandwidth consumption is $G_t=\min\left\{X_t,qS_t/\Theta^{\max}\right\}.$
Thus the state \(Z_t=(X_t,B_t)\) is a time-homogeneous Markov chain.

\smallskip\noindent{\bf Fluid Approximation.}
To analyze the large-market behavior, we replace the random eligible demand $S_t$ by its conditional mean before applying the capacity constraint. The resulting fluid consumption is $g(x,p):=\min\{x,q\lambda(1-F(p))\},$
and the deterministic fluid dynamics are
\[x_{t+1} = x_t - g(x_t, b_t) + \alpha(1 - x_t), \qquad
b_{t+1} =
\begin{cases}
  b_t\!\left(1 + d \cdot \frac{g(x_t, b_t) - x_t/2}{x_t/2}\right)
    & \text{if } x_t > 0, \\
  b_t & \text{if } x_t = 0.
\end{cases}\]
This approximation is asymptotically accurate in large markets. The bias comes from applying the concave truncation \(\min\) before expectation.

\smallskip
\noindent{\it Fluid fixed point.}
The fluid system has an operating point where consumption equals the target ($g(x^*,b^*)=x^*/2$) and recovered bandwidth exactly replenishes the amount consumed ($\alpha(1-x^*)=g(x^*,b^*)$). 
Solving them gives
\[
x^*=\frac{2\alpha}{2\alpha+1},
\qquad
F(b^*)=
1-\frac{\alpha}{q\lambda(2\alpha+1)}.
\]
Thus, whenever the arrival rate is large enough to sustain a positive price, the system has a unique positive-price fixed point: supply stabilizes at $x^*$, and the price $b^*$ is the value threshold that induces the correct eligible demand. 
We develop this formally in the full version.

\smallskip
\noindent{\it Local Stability.}
It remains to show that this operating point is stable rather than merely stationary. Around the fixed point, the capacity constraint is slack, so the fluid map is differentiable. 
We linearize the two-dimensional update in relative supply and price deviations and apply the Jury stability criterion for discrete-time systems. 

This yields an explicit upper bound on the step size $d$, based on the recovery rate $\alpha$ and the local demand elasticity at $b^*$, such that every sufficiently small perturbation decays back to the fixed point. If the step size is too large, the linearized system has spectral radius exceeding one, and the price-supply dynamics locally oscillate away from the fixed point. 
Combined with a finite-horizon tracking result, this shows that in large markets the stochastic process $Z_t$, started near the fixed point, remains near it over any fixed horizon with high probability, with deviations from the fluid trajectory of order \(1/\sqrt{\Theta^{\max}}\).

%% file: Sections/simple-model.tex
\section{Background and Preliminaries}
\label{sec:background}

\paragraph{Notations.} We use $[n]$ to denote the set of integers $\{1,...,n\}$. 
Given a vector ${\bf b}$, we use ${\bf b}[i,j]$ to denote the sub-vector of ${\bf b}$ from the $i$-th to $j$-th entry, both ends included.
We use ${\bf 1}_{E}$ to represent the indicator of whether an event $E$ happens.

\subsection{Consensus and Proof-of-Availability}
The bandwidth fee mechanism in this work does not depend on the internal details of a particular consensus protocol, and thus is applicable to the family of split layer consensus protocols. In this work the \emph{proof-of-availability layer} certifies that transaction data has been disseminated to enough validators to make the data retrievable. The \emph{consensus layer} consumes availability certified proposals and decides which certified transactions become part of the finalized ledger.

Fix a validator set $V=[m]$ and an adversarial set $F\subseteq V$ with $|F|\le \faults$. Validators have public verification keys, signatures are unforgeable, and honest validators follow the attestation rule stated below. We work in an authenticated partially synchronous network~\cite{DLS88}, meaning there is an unknown global stabilization time (GST) and a finite delay bound $\Delta$ such that every message sent between honest validators after GST is delivered within $\Delta$.  Before GST, message delays are controlled by the adversary.

The safety claims of the PoA layer do not rely on GST or on a known network bound. They rely only on unforgeability and on the fact that honest validators sign availability statements only after receiving and storing the corresponding data. The post-GST delivery assumption is used only for liveness-style claims: retrieving data from an honest attester, collecting signatures from honest attesters, and the consensus layer's inclusion liveness. A synchronous network is the special case GST$=0$ with a known delivery bound, and is not needed for the baseline centralized planner model unless the mechanism also introduces deadline based penalties for missed assigned work.

The proof-of-availability threshold is an integer $k$ satisfying
\[
    \faults < k \le m-1.
\]
The common choice $k=\faults+1$ is enough for availability soundness where every valid
certificate contains at least one honest attester. Stronger quorum intersection
properties, if needed, belong to the consensus layer and should not be inferred
from the PoA certificate alone.

\begin{definition}[PoA certificate]
\label{def:poa-certificate}
A proposal header is a tuple
\[
    h=(s,\ell,\mathsf{parent},D),
\]
where $s\in V$ is the source validator, $\ell\in\mathbb{N}$ is the source lane position, $\mathsf{parent}$ is the parent proposal identifier, and $D$ commits to the proposal payload.  For a transaction $\tx$ contained in that payload, a PoA certificate is a tuple
\[
    \pi=(h,\tx,A,\{\sigma_j\}_{j\in A}),
\]
where $A\subseteq V\setminus\{s\}$ and $|A|\ge k$.  The certificate is valid, written $\mathsf{ValidPoA}(\pi)=1$, if for every $j\in A$,
\[
    \mathsf{Verify}_{j}\!\bigl(\sigma_j,
    (\mathsf{PoA},h,H(\tx))\bigr)=1.
\]
The validators in $A$ are called the \emph{attesters} for $\tx$ under header $h$.
\end{definition}

\begin{definition}[PoA certified proposal]
\label{def:poa-certified-proposal}
A proposal
\[
    P=(h,T,\{\pi_{\tx}\}_{\tx\in T})
\]
with source $s$ is PoA certified if every transaction $\tx\in T$ has a valid PoA certificate $\pi_{\tx}$ for the same header $h$.  We write $\mathsf{ValidProp}(P)=1$ for this predicate.
\end{definition}

An honest validator signs $(\mathsf{PoA},h,H(\tx))$ only after receiving the transaction data committed by $h$ and storing it for the data availability window. Moreover, for each source lane position $(s,\ell)$, an honest validator signs messages for at most one proposal header $h$. This rule makes honest attestations meaningful as availability evidence, it does not by itself make
PoA a consensus certificate.

\begin{definition}[PoA layer guarantees]
\label{def:poa-layer-guarantees}
The PoA layer provides the following guarantees.
\begin{enumerate}
    \item \emph{Public verifiability.} Anyone can check
    $\mathsf{ValidPoA}(\pi)$ and $\mathsf{ValidProp}(P)$ using public keys and
    the transaction data.

    \item \emph{Availability soundness.} If $\mathsf{ValidPoA}(\pi)=1$ and
    $k>\faults$, then at least one honest validator attested to the corresponding
    transaction.  Hence, after GST and during the availability window, every
    honest validator can retrieve the transaction from some honest attester.

    \item \emph{Accountability of work.} A valid certificate publicly identifies
    the source validator and the actual attesters.  Thus the mechanism can base
    rewards on observed certified work rather than on the planned route alone.

    \item \emph{Honest completion.} If, after GST, an honest source sends a
    transaction to a set containing at least $k$ honest attesters with sufficient
    available bandwidth, then those attesters eventually return enough valid PoA
    signatures to certify the transaction.
\end{enumerate}
\end{definition}

\begin{lemma}[Soundness of the relaxed PoA interface]
\label{lem:relaxed-poa-soundness}
Assume $|F|\le \faults$, $k>\faults$, unforgeable signatures, and the honest attestation
rule.  Then the PoA relaxation used in this paper satisfies the availability
claims in \Cref{def:poa-layer-guarantees}, even when different transactions in
the same proposal use different attester sets and a source sends a transaction
only to the attesters used in that transaction's certificate. In particular every valid PoA certificate for a transaction identifies at least one honest validator that received and stored that transaction during the availability window, and after GST, if an honest source sends a transaction to a set containing at least $k$ honest attesters with sufficient available bandwidth, then the source eventually obtains a valid PoA certificate for that transaction.
\end{lemma}

\begin{proof}[Sketch]
Let $\pi=(h,\tx,A,\{\sigma_j\}_{j\in A})$ be a valid PoA certificate. Since $|A|\ge k>\faults$ and at most $t$ validators are Byzantine, there exists some honest validator $j\in A\setminus F$.  By unforgeability, $\sigma_j$ was produced by $j$. By the honest attestation rule, $j$ signs $(\mathsf{PoA},h,H(\tx))$ only after receiving the transaction data committed by $h$ and storing it for the availability window. Thus the certificate contains an honest data holder. If an honest validator queries $j$ after GST while the availability window is still open, the request and $j$'s response are delivered within the post-GST network bound, so the transaction is retrievable. For liveness, suppose an honest source sends the transaction after GST to a set containing at least $k$ honest attesters and those attesters have sufficient available bandwidth. Partial synchrony implies that each honest attester in the set eventually receives the data.  By the honest attestation rule, each stores the data and returns a valid signature, and the source can assemble $k$ such signatures into a valid PoA certificate. The proof is per transaction and therefore does not require a common attester set for all transactions in a proposal, nor does it require all-to-all broadcast. It proves availability only; agreement, lane consistency, and finality are the responsibilities of the consensus layer in \Cref{def:consensus-layer}.
\end{proof}

\begin{definition}[Consensus layer]
\label{def:consensus-layer}
A consensus layer for PoA certified proposals is an abstract functionality that receives submitted proposals $P$ and outputs an increasing sequence of finalized ledger prefixes
\[
    \mathcal{L}_0 \preceq \mathcal{L}_1 \preceq \cdots .
\]
It satisfies:
\begin{enumerate}
    \item \emph{Agreement.} Any two honest validators' finalized
    ledgers are prefix comparable, and a finalized transaction is never removed.

    \item \emph{PoA validity.} Every finalized transaction appears in some
    proposal $P$ with $\mathsf{ValidProp}(P)=1$.

    \item \emph{Lane consistency.} For each source validator $s$, the finalized
    proposals from lane $s$ form a chain: at most one proposal is finalized for
    each pair $(s,\ell)$, and finalized proposals respect parent pointers.

    \item \emph{Inclusion liveness.} After GST, if an honest validator submits a
    valid PoA-certified proposal that is non-conflicting and whose transactions
    are not already finalized, then the proposal is eventually included in the
    finalized ledger.
\end{enumerate}
\end{definition}

A transaction is \emph{committed} only when it appears in the finalized ledger
of the consensus layer.  Being PoA certified means that the transaction is
available and admissible for consensus, it is not the same as being committed.
If the same transaction appears in several certified proposals, the ledger counts
only its first finalized occurrence for payment and reward accounting.

For a finalized ledger slice $L$, let ${\sf src}(\tx)$ denote the source validator of transaction $\tx$, and let ${\sf att}(\tx)$ denote the attester set appearing in the PoA certificate of the first finalized occurrence of $\tx$.  Define
\[
    B_i(L) := \#\{\tx\in L : {\sf src}(\tx)=i\}
\]
and
\[
    A_i(L) := \#\{\tx\in L : i\in {\sf att}(\tx)\}.
\]
The verifiable bandwidth work of validator $i$ in $L$ is
\[
    \beta_i(L) := (k+1)B_i(L)+A_i(L).
\]
The term $(k+1)B_i(L)$ accounts for receiving each sourced transaction from the planner and sending it to $k$ attesters in the exactly $k$ baseline; $A_i(L)$ accounts for the transactions that $i$ actually attested for other sources.  If a variant uses redundant fanout $q>k$ for liveness, then either the source term must be changed to $(q+1)B_i(L)$ with a verifiable receipt rule, or rewards must explicitly cover only the certified sends.  Payments should be charged only to committed transactions, and validator rewards should be functions of $\beta_i(L)$, not merely of the ex ante assignment $\mathcal{R}$, unless the mechanism also includes a separate accountability/slashing rule for missed assigned attestations.

The baseline mechanism is non-slashing. A validator that refuses to sign an assigned transaction is simply absent from the resulting PoA certificate and receives no reward for that attestation. If a source cannot collect enough signatures, the corresponding transaction is not counted as committed and is not charged. 
Thus the baseline incentive analysis does not need to distinguish a strategic refusal to attest from a pre-GST network delay.

In the myopic mechanism, we can use the following simplifying assumption. After GST and by the end of the accounting round, the consensus layer finalizes every non-conflicting PoA certified proposal produced by an honest source, and finalizes no proposal that fails $\mathsf{ValidProp}$. This is an assumption about the consensus functionality, not a consequence of the PoA certificate alone. Under this assumption, the mechanism can reason about the finalized ledger slice without modeling the internal consensus messages. When convenient, we identify that ledger slice with the tuple of finalized source lane proposals $\block_1,\ldots,\block_m$.

\subsection{Preliminary on Markov Process}

\begin{definition}[Markov Process]
A stochastic process $\{Z_t\}_{t \geq 0}$ is Markov if for all $t \geq 0$ and all states $z$, for all state $z_0,...,z_t$ such that $\Pr[Z_i = z_i \text{ for }i=0,...,t]>0$, we have
\[
    \Pr[Z_{t+1} = z \mid Z_0=z_0, \ldots, Z_t=z_t] = \Pr[Z_{t+1} = z \mid Z_t=z_t].
\]
That is, the future depends only on the current state, not on the history.
\end{definition}

\begin{definition}[Time-homogeneous]
A Markov chain is time-homogeneous if its transition probabilities do not depend on the time step $t$. For any state $z$ and $z'$ in the support,
\[
    \Pr[Z_{t+1} = z' \mid Z_t = z] = \Pr[Z_1 = z' \mid Z_0 = z], \forall t
\]
\end{definition}

\begin{definition}[Fixed Point]
\label{def:fixedpoint}
A state $z^* \in \mathcal{S}$ is a fixed point of a discrete-time dynamical system $z_{t+1} = H(z_t)$ if $H(z^*) = z^*$, i.e., applying the update rule leaves the state unchanged. 
\end{definition}

\begin{definition}[Local stability]
\label{def:local-stability}
A fixed point \(z^*\) of \(z_{t+1}=H(z_t)\) is locally asymptotically stable if there exists a neighborhood \(U\) of \(z^*\) such that every trajectory starting in \(U\) remains near \(z^*\) and converges to \(z^*\).
\end{definition}

\section{Models and Definitions}
\label{sec:models-definitions}
\subsection{Single-Round Bandwidth Fee Game with Centralized Planner}
Unlike the traditional leader-based consensus, 
in decoupled dissemination type of consensus, transactions first compete for validators' bandwidth of send/receive messages and generate PoA. 
However, validator's bandwidth capacity and their cost of spending the bandwidths are not directly observable by the protocol itself; therefore, we will need to rely on self-reports from the validators.

In this work, we focus on a \emph{central planner model}: A trusted planner receives user transactions and validators' reports of bandwidth capacity and costs and assigns users' transactions to routes.
We start with a single-round bandwidth fee game where we discretize the time into rounds and focus on the ledger slice generated in this round.

\paragraph{Players.} 
We consider \emph{myopic} players, who maximize their own immediate utility in a single round and ignore the effect of their actions on future gains.
There are two types of roles: \emph{users} who wish to have their transactions committed, and a set of $m$ \emph{validators} $V = \{1,...,m\}$ who spend their bandwidth to disseminate users' transactions in exchange for payments. 
For simplicity, we consider a normalization and assume that each transaction has the same size and costs one unit of bandwidth. 
Moreover, signatures are relatively small in size and we ignore the bandwidth cost of sending/receiving only signatures.  

Each user $u$ has a private valuation $b^*_u \ge 0$ that measures the maximum amount they are willing to pay to get their transactions committed.
Each validator $i\in[m]$ has a private bandwidth capacity $\band^*_i$ and a private cost function $f^*_i: [0,\band^*_i] \rightarrow \mathbb{R}_{\ge 0}$ that measures how much it costs validator $i$ to use $x$ units of bandwidth. 
We assume that each validator's cost function is a single-parameter cost $f^*_i(x) = c_i\cdot g_i(x)$ where $g_i(x)$ is the ``shape'' of the cost function and is \emph{publicly known}, while the scalar $c_i > 0$ is private and is only known to validator $i$. 
For simplicity, we assume that the shape function $g_i(x)$ is continuous and strictly increasing over $[0,\band_i]$. Moreover, $g_i(0) = 0$.

The function $g_i$ describes the shape of how validator $i$'s total resource usage grows with its workload, while the private scalar $c_i$ converts this resource usage into the validator's local utility cost. The canonical choice $g_i(x) = x$ gives a linear cost, in which $c_i$ is validator $i$'s private per-unit cost, such as its electricity price or opportunity cost. This arises, for example, when a validator's hardware consumes approximately a fixed amount of incremental energy per transmitted byte. 
More generally, $g_i$ can be an externally specified or benchmarked curve associated with a hardware or network class, for instance, one published by the protocol for each class or measured on reference hardware. Different hardware classes may therefore have different public shapes $g_i$, while $c_i$ captures validator-specific local costs such as the price of electricity or bandwidth where the validator operates. Our incentive analysis thus treats $g_i$ as public.

\paragraph{Bandwidth Fee Mechanism Game.}
Let $\C$ denote a strategic individual or coalition. 
Specifically, $\C$ can be a user, (a set of) validators, or a coalition of some validators and one or more users. 
Coalition members share all information, including their private types and private signing keys.
Below, we may use bid or transaction interchangeably.
A bandwidth fee mechanism (BFM) describes the following game:
\begin{enumerate}
    \item Honest validators report their private cost functions (or equivalently, the private scalar $c_i$) and capacity. 
    \item Validators in $\C$, if any, submit their private costs and capacity based on honest validators' reports. Denote the reported cost functions of all validators as $\boldsymbol{f} = (f_1,...,f_m)$ and bandwidths $\boldsymbol{\band} = (\band_1,...,\band_m)$. 
    \item Honest users submit bids to the centralized planner, each represented by a single non-negative real value.
    \item The coalition $\C$ submits its bids to the centralized planner based on validators' reported cost functions $\boldsymbol{f}$ and capacities $\boldsymbol{\band}$, and honest users bids. 
    \item The planner sends to each validator $i$ the transactions allocated to $i$, and which $k$ other validators to attest the transaction. 
    \item Each validator forms a proposal and interact with other validators.
    At the end of the game, each validator $i\in[m]$ produces a finalized source lane proposal $\block_i$. 
\end{enumerate}
To specify a bandwidth fee mechanism in this model, we only need to specify the following rules which can be possibly randomized algorithms.
Given any reported cost functions $\boldsymbol{f}$ and capacities $\boldsymbol{\band}$ from the validators and users' bids
$\boldsymbol{b}$, and the finalized proposals $\block_1,...,\block_m$ produced at the end of this round, a BFM consists of:
\begin{itemize}
    \item {\bf Allocation} rule $\mathcal{R}\leftarrow \assign(\boldsymbol{f},\boldsymbol{\band}, \boldsymbol{b})$ specifies the assigned route for each transaction.
    Here, $\mathcal{R} = (R(b_1),...,R(b_n))$, where each $R(b) = (s(b);A(b))$ represent a route $(s;A)$ for the transaction represented by bid $b$.
    Here, $\{s\}\cup A\subseteq [m]\cup\{\bot\}$, and $s(b) = \bot$ means a transaction will not be disseminated this round.

    \item {\bf Payment} rule $(p_1,...,p_n) \leftarrow \price(\mathcal{R},\block_1,...,\block_m)$ outputs the non-negative payment charged to each user. For simplicity, we denote $p_u = \price_u(\block_1,...,\block_m)$ for $u\in[n]$. 
    \item {\bf Reward} rule $(r_1,...,r_m) \leftarrow \reward(\mathcal{R}, \block_1,...,\block_m)$ outputs the reward paid to the validators. For simplicity, we use $r_i = \reward_i(\mathcal{R}, \block_1,...,\block_m)$.
    We allow the reward to be negative to capture a general model that supports penalty, although our mechanism achieves incentive compatibility without slashing.
\end{itemize}
We assume that the (honest) BFM achieves the following properties: 
\begin{itemize}
    \item \emph{Attestation feasibility:} 
    Each transaction gets assigned at most one route.
    If $s = \bot$, then all $a_1,...,a_k = \bot$; otherwise if $s\in[m]$, then $A\subset[m]$.
    \item \emph{Individual rationality}: a user $u$ only pays if its transaction is committed, and the payment should not exceed its bid, i.e., $p_u \leq b_u$; moreover, a validator's reward should cover its cost on \emph{assigned} bandwidth usage based on the reports, i.e., $r_i\geq f_i(\beta^{\rm ass}_i)$, where $\beta^{\rm ass}_i = \sum_{R\in\mathcal{R}} w_i(R)$ is the assigned bandwidth usage of validator $i$.
    \[
    w_i(R) := 
    \begin{cases} 
        k+1 & \text{if } i = s, \\ 1 & \text{if } i \in A, \\ 0 & \text{otherwise.}
    \end{cases}
    \]
    \item \emph{Budget feasibility}: the total validator rewards cannot exceed the total user payments. Any surplus is burnt.
\end{itemize}

\paragraph{Strategies.}
We focus on direct-revelation mechanisms. 
For a user with true value $b^*$, the honest strategy is to submit a single bid representing its true value.
A strategic user $u$ may choose to submit a bid vector $\boldsymbol{b}_u$ that contains zero to multiple bids, which do not necessarily reflect their true values. We call all the additional bids that the user injects as \emph{fake bids}, and we equivalently view the true value of the fake bids as zero.
In practice, a user may also bypass the planner and send its transaction privately to some validators. 
However, these transactions do not carry the planner's signature, so validators will not get reward for routing these transactions.
Therefore, in the rest of the paper, we would ignore such strategies without loss of generality.

An honest validator $i$ reports its true private cost function and capacity, taking the assigned user transactions for which it is the source, and send the transactions to corresponding attesters according to $\mathcal{R}$.
A strategic validator may misreport its cost function and capacity, inject fake bids as a user, does not follow $\mathcal{R}$ for attestation, or refuse to provide attestation to other validators. 
We do note that when an attester receives the transaction from the source, the receive bandwidth cost is incurred whether or not it later returns a signature.

A coalition shares all the information and coordinate the strategies together. 
A coalition's strategy space is the combination of all members' strategies.
Importantly, if there are multiple validators in the coalition, one can sign transactions on behalf of other validators in the coalition using their signing keys.
\paragraph{Utilities.}
A user $u$ with true value $b^*_u$ gets utility 
\[b^*_u\cdot {\bf 1}_{\{u\text{'s real transaction finalized}\}}-\sum_{\text{all bids $x$ submitted by }u}p_x,\]
A validator $i$ gets utility
\[r_i - f^*_i(\beta^{\rm act}_i)-\sum_{\text{all fake bids $x$ submitted by }i}p_x\]
if it gets reward $r_i$ and spent $\beta^{\rm act}_i$ unit of bandwidth in the actual work. 
If a validator runs out of bandwidth, it crashes and gets $-\infty$ utility. A coalition's utility is the sum of utility of all coalition members.

\subsection{Incentive Compatibilities}
Fix a strategic coalition (or individual) $\C$.
Let $H_{\mathcal{C}}$ denote the coalition's honest strategy. Let $\boldsymbol{b}^*, \boldsymbol{\band}^*, \boldsymbol{f}^*$ be the true valuation of users, bandwidth capacity, and cost function of validators, respectively. 
Let $\text{util}_{\C}(\boldsymbol{b}^*, \boldsymbol{\band}^*, \boldsymbol{f}^*; \mathcal{S}_{\mathcal{C}})$ denote the expected utility of $\C$ in a randomized experiment where players' true types are represented by $\boldsymbol{b}^*, \boldsymbol{\band}^*, \boldsymbol{f}^*$, and players in $\mathcal{C}$ jointly adopt strategy $\mathcal{S}_{\mathcal{C}}$ while all other players act honestly. 
The expectation is taken over the randomness in strategy $\mathcal{S}_{\mathcal{C}}$ and in the mechanism.
\begin{definition}[Incentive Compatibility]
\label{defn:ic}
    Let $\C$ be arbitrary strategic coalitions or individual.
    We say that a BFM is incentive compatible against a coalition $\mathcal{C}$ if for any true types $\boldsymbol{b}^*, \boldsymbol{\band}^*, \boldsymbol{f}^*$, for any strategy $S_{\mathcal{C}}$ of coalition $\mathcal{C}$, we have:
\[
{\sf util}_{\mathcal{C}}(\boldsymbol{b}^*, \boldsymbol{\band}^*, \boldsymbol{f}^*; \mathcal{S}_{\mathcal{C}}) \leq {\sf util}_{\mathcal{C}}(\boldsymbol{b}^*, \boldsymbol{\band}^*, \boldsymbol{f}^*; H_{\mathcal{C}}).
\]
\end{definition}
\noindent Specifically, we say that a BFM is 
\begin{itemize}
    \item \emph{User incentive compatible (UIC)} if \Cref{defn:ic} holds for any $\C$ containing an individual user.
    \item $c$-\emph{validator incentive compatible ($c$-VIC)} if \Cref{defn:ic} holds for any for any $\C$ containing up to $c$ number of validators. 
    \item \emph{$(c,d)$-side contract proofness ($(c,d)$-SCP)} if \Cref{defn:ic} holds for any $\C$ containing at least one and up to $c$ number validators and at least one and at most $d$ users.
    When there is only one validator in the coalition, we omit the $c$ parameter and refer to $d$-SCP.
\end{itemize}

%% file: Sections/impossibility.tex
\section{Impossibility of Coalition Resilience}
\label{sec:impossibility}
In this section, we present the barriers for BFM against coalitions.
First, once validators can delegate signing keys within a coalition, and the mechanism cannot certify that the intended data transfer actually happened. In this case, coalition of validators can save their bandwidth costs by signing on behalf of other coalition members.
Second, even after ruling out coalitions with two or more validators and focus on coalitions with one validator with users, we still suffer from restrictions due to the competition for finite amount of bandwidth.
\subsection{Validator-Coalition Impossibility under Delegated Signing.}

We start with the impossibility against validator coalitions. The impossibility is broad in the sense that if rewards are computed only from the public transcript, then two validators can internally simulate the transcript of an honest transfer while avoiding the underlying bandwidth cost.

We say that a BFM is \emph{transcript-based} if, after the reports and bids are fixed, payments and rewards depend only on the planner assignment $\mathcal{R}$ and the finalized PoA-certified proposals $(\block_1,...,\block_m)$.
In particular, the mechanism cannot observe which physical validator or IP address actually sent, received, or stored transaction data except through the public signatures appearing in the finalized PoA certificates. 

This abstraction captures the usual permissionless blockchain model, where a protocol can only base rewards and slashing on information that is deterministically available to all validators such as bids, assignments, blocks, hashes, signatures, and certificates. By contrast, facts such as which physical machine sent a packet, which peer received the full transaction body, or which disk stored the data are not part of the replicated ledger state.  Network metadata such as IP addresses, connection logs, timing information, or peer-to-peer routing paths is local, forgeable or relayable, and generally not consensus verifiable. Moreover, in practice a validator identity is a public key rather than a physical machine as validators often use remote signers, sentry nodes, load balancers, custody services, or threshold signing infrastructure.  Thus, absent additional trusted hardware or a separate verifiable data possession primitive, the protocol cannot distinguish ``validator \(a\) received the data and signed'' from ``another coalition member obtained signing authority for \(a\)'s key and produced the same signature.'' Accordingly, any reward rule implementable by the base protocol is naturally a
function of the public transcript. We therefore model coalitions as being able to delegate signing authority internally such that a coalition member may either share its signing key, run a threshold signing protocol, or give another coalition member access to a remote signer. Since all these cases produce the same public signature, a transcript based mechanism cannot distinguish
them.

\begin{theorem}
\label{thm:no-2vic}
No non-trivial transcript-based BFM is $2$-VIC.
\end{theorem}

\begin{proof}
Suppose, for the sake of contradiction, that $\mathcal{M}$ is a $2$-VIC non-trivial transcript-based BFM. Then there exists validators' types $(\boldsymbol f,\boldsymbol{\band})$ and users' bids $\boldsymbol b$, such that some transaction ${\sf tx}$ in $\boldsymbol b$ is committed with a positive probability. Then there is an event $E$ such that some source $s$ and attester set $A$ is assigned for ${\sf tx}$ and ${\sf tx}$ is committed in an honest execution.

Consider a validator coalition $\C = \{s,a\}$ for some $a\in A$.
In the honest case, $s$ sends the full transaction $\tx$ to $a$, after which $a$ signs $(\mathsf{PoA},h,H(tx))$ with a valid signature $\sigma_a$.
Now imagine a strategy $S_{\C}$ as follows:
If $E$ happens, the validator $s$ does not send the transaction, but instead, locally computes $\sigma_a$ using signing key of $a$. Otherwise, the coalition behaves honestly. 

Under this strategy, the finalized PoA-certified proposals $(\block_1,...,\block_m)$ is the same as in honest case.
Since $\mathcal{M}$ is transcript-based, the public transcript relevant to payments and rewards is identical as in the honest execution and in the deviating execution under $S_{\C}$. Therefore, all validator rewards are the same in the two executions.

Let $\beta_s^H,\beta_a^H$ denote the actual bandwidth usages of $s$ and $a$ in the honest execution, and let $\beta_s^S,\beta_a^S$ denote their bandwidth usages under strategy $S$. Then conditioned on event $E$, we have
\begin{align*}
    &\beta_s^S=\beta_s^H-1, \qquad \beta_a^S=\beta_a^H-1,\\
    &r_s^S=r_s^H, \qquad r_a^S=r_a^H.
\end{align*}
Therefore, conditioned on event $E$, the coalition's utility under strategy $S_{\C}$ is strictly higher than the honest utility by the strict monotonicity of cost functions. 
Since event $E$ happens with a positive probability, the coalition has a strictly profitable deviation, contradicting 2-VIC.
\end{proof}
The implication of \Cref{thm:no-2vic} is that any non-trivial design whose rewards depend only on publicly available transcripts cannot hope for $2$-VIC against validator coalitions.  
To circumvent this impossibility one would need additional assumptions/tools that either prevents key sharing within the coalition, or design reward rule based on some physical data receipts that cannot be delegated.

\begin{corollary}
\label{cor:no-2scp}
No non-trivial transcript-based BFM is $(2,d)$-SCP for any $d\geq 1$.    
\end{corollary}
\begin{proof}
    Suppose, for the sake of contradiction, that $\mathcal{M}$ is a $(2,d)$-SCP non-trivial transcript-based BFM for some $d\geq 1$. Then there exists validators' types $(\boldsymbol f,\boldsymbol{\band})$ and users' bids $\boldsymbol b$, such that some transaction ${\sf tx}$ in $\boldsymbol b$ is committed with a positive probability. 
    Thus, there exist an event $E$ such that some source $s$ and attester set $A$ is assigned for ${\sf tx}$. Consider a coalition of validators $s$ and $a$, and $d$ number of users who adopt the same delegate signing strategy as in \Cref{thm:no-2vic}.
    Note that this strategy does not change the distribution of committed transactions, and therefore does not change the colluding users' expected utility.
    By a similar argument as in \Cref{thm:no-2vic}, the coalition has a strictly profitable deviations, contradicting $(2,d)$-SCP.
\end{proof}

\subsection{Impossibility under Validator-User Coalitions}
Due to Theorem ~\ref{thm:no-2vic} and Corollary ~\ref{cor:no-2scp} in the remainder of the impossibility analysis we focus on validator-user coalitions that involve only one validator. We prove this impossibility through a reduction to the MPC-assisted TFMs proposed in ~\cite{shi2023_tfm}.

In the MPC-assisted model, a TFM auctions off the finite block space to users. In each round at most $\Theta$ transactions can be included in the block, where $\Theta$ is the block size. 
An ideal functionality $\mathcal{F}_{\mathrm{MPC}}$ receives transactions, determines the allocation, users' payments, and validators' net revenue from the mechanism. For completeness we give a more formal description of the MPC-assisted model in ~\cref{app:mpc-assisted-tfm}.

At a high level, the reduction works as follows. We fix validators' true types $(\boldsymbol{\band}^*, \boldsymbol{f}^*)$ and assume that all validators report honestly and behave honestly during the dissemination.
The resulting BFM can thus be viewed as a finite block MPC-assisted TFM with $\Theta = \sum_i \mathrm{cap}_i$: 
Because the strategy space and the utility in both the induced MPC-assisted TFM and the BFM are the same, so the incentive compatibility of BFM transfers to the induced
MPC-assisted TFM.

\begin{lemma}
\label{lem:reduction}
Fix a non-trivial BFM $\mathcal{M} = (\assign,\price,\reward)$ in the planner model. There exists a non-trivial MPC-assisted TFM $\overline{\mathcal{M}}$ such that for any coalition or individual $\C$, if $\mathcal{M}$ is incentive compatible against $\C$, then $\overline{\mathcal{M}}$ is incentive compatible against $\C$.
\end{lemma}

\begin{proof}
Fix an arbitrary validator type profile $(\boldsymbol{\band}^*,\boldsymbol f^*)$ such that under truthful validator reports and honest dissemination, there
exists a bid vector $\boldsymbol{b}$ for which some bid is confirmed with positive probability.
Throughout the proof, we fix the validators' reports in $\mathcal M$ to be
$(\boldsymbol{\band}^*,\boldsymbol f^*)$, and we assume that validators execute the assigned dissemination honestly.

Let $\Theta := \sum_{i\in[m]} \band_i^*$.
We now construct an MPC-assisted TFM $\overline{\mathcal M}$ with the $m$ validators in $V$ for the block-space setting with block size $\Theta$.
On input bid vector $\boldsymbol b$, the ideal functionality
$\mathcal F_{\rm MPC}$ computes $\mathcal{R}
        \leftarrow
    \assign(\boldsymbol f^*,\boldsymbol{\band}^*,\boldsymbol b)$.
Let $\boldsymbol B^H
        =
    (\block_1,\ldots,\block_m)$
denote the finalized source-lane proposals when all validators execute
the assignment $\mathcal{R}$ honestly. Note that under honest dissemination, the committed transactions are exactly the transactions with $s(\tx)\neq \bot$.  The MPC-assisted TFM $\overline{\mathcal{M}}$ includes all transactions with non-bot routes. 
We use $\boldsymbol X$ to denote the set of included transactions $\overline{\mathcal{M}}$.
By individual rationality of the validators in $\mathcal{M}$, we have $|\boldsymbol X|\leq \Theta$.
Thus the block-size constraint is satisfied.

The user payments of $\overline{\mathcal M}$ are defined by $(p_1,\ldots,p_n)
        \leftarrow
    \price(\mathcal{R}, \boldsymbol B^H)$.
The validator revenues $(\overline r_1,\ldots,\overline r_m)$ of $\overline{\mathcal M}$ are defined as the net validator utility in the honest dissemination execution: 
\[\overline r_i:= r_i - f_i^*(\beta_i^{\rm ass}),\]
where $(r_1,\ldots,r_m)\leftarrow\reward(\mathcal{R},\boldsymbol B^H)$, and $\beta_i^{\rm ass} :=\sum_{R\in \mathcal{R}} w_i(R)$ is the assigned bandwidth usage for validator $i$ in $\mathcal{M}$.

The induced MPC-assisted TFM is well-defined. Individual rationality and budget feasibility follows from that of $\mathcal M$. Moreover, the validator's revenue is non-negative in $\overline{\mathcal M}$ by individual rationality.
Finally, $\overline{\mathcal M}$ is non-trivial by the choice of
$(\boldsymbol{\band}^*,\boldsymbol f^*)$.

Now fix a coalition $\mathcal C$.  In the MPC-assisted model, users may
strategically bid, and validators may inject fake bids.  Therefore, for any strategy $\overline S_{\mathcal C}$ admissible for $\mathcal C$ in $\overline{\mathcal M}$, there is a corresponding strategy $S_{\mathcal C}$ in $\mathcal M$ defined as follows:  The validator members of $\mathcal C$, if any, report their private types honestly. The coalition submits exactly the same strategic bid vector as prescribed by $\overline S_{\mathcal C}$.  
All validators then execute the dissemination honestly.

By construction, for every bid vector $\boldsymbol b$ and every realization of
the mechanism's randomness, the two executions produce the same committed transactions, the same user payments, and the same validator utility.
Therefore,
\[
    {\sf util}_{\mathcal C}^{\mathcal M}
        (\boldsymbol{b}, \boldsymbol{\band}^*, \boldsymbol{f}^*; S_{\mathcal C})
    =
    {\sf util}_{\mathcal C}^{\overline{\mathcal M}}
        (\Theta,\boldsymbol b;\overline S_{\mathcal C}),
\]
where ${\sf util}_{\mathcal C}^{\overline{\mathcal M}}
        (\Theta,\boldsymbol b;\overline S_{\mathcal C})$ denotes $\C$'s expected utility in $\overline{M}$ when $\C$ adopts strategy $\overline{S}_{\C}$ on input bid vector $\boldsymbol{b}$.
The same equality holds when $\mathcal C$ follows the honest strategy.

Suppose $\mathcal M$ is incentive compatible against $\mathcal C$.  If
$\overline{\mathcal M}$ were not incentive compatible against $\mathcal C$, then
there would exist a true value vector $\boldsymbol b$ and a strategy
$\overline S_{\mathcal C}$ such that
\[
    {\sf util}_{\mathcal C}^{\overline{\mathcal M}}
        (\Theta,\boldsymbol b;\overline S_{\mathcal C})
    >
    {\sf util}_{\mathcal C}^{\overline{\mathcal M}}
        (\Theta,\boldsymbol b;H_{\mathcal C}).
\]
By the utility equality above, the corresponding strategy $S_{\mathcal C}$ would
satisfy
\[
    {\sf util}_{\mathcal C}^{\mathcal M}
        (\boldsymbol{b}, \boldsymbol{\band}^*, \boldsymbol{f}^*;S_{\mathcal C})
    >
    {\sf util}_{\mathcal C}^{\mathcal M}
        (\boldsymbol{b}, \boldsymbol{\band}^*, \boldsymbol{f}^*;H_{\mathcal C}),
\]
contradicting incentive compatibility of $\mathcal M$ against $\mathcal C$.
Hence $\overline{\mathcal M}$ is incentive compatible against $\mathcal C$.
\end{proof}

\begin{theorem}
\label{thm:data-pricing-uic-2scp-impossible}
No non-trivial BFM $\mathcal{M} = (\assign,\price,\reward)$ in the planner model achieves UIC, $1$-VIC and $1$-SCP.
\end{theorem}

\begin{proof}
Suppose for the sake of contradiction that a non-trivial BFM $\mathcal{M}$ that achieves UIC, $1$-VIC, and $1$-SCP.
By \Cref{lem:reduction}, for any fixed validator's true types $(\boldsymbol{\band}^*,\boldsymbol f^*)$, the BFM induces an MPC-assisted TFM that satisfies UIC, $1$-VIC, and $1$-SCP.
By \Cref{thm:mpc-zero-validator-revenue}, we have that for any validator's true types $(\boldsymbol{\band}^*,\boldsymbol f^*)$, for any bid vector $\boldsymbol{b}^*$, in the honest execution, validator's utility in BFM must always be zero:
\[\util_i^{\mathcal{M}}(\boldsymbol{b}^*, \boldsymbol{\band}^*, \boldsymbol{f}^*) = 0.\]

Now fix an arbitrary validator type profile $(\boldsymbol{\band}^*,\boldsymbol f^*)$ such that under truthful validator reports and honest dissemination, there
exists a bid vector $\boldsymbol{b}$ for which some bid is confirmed with positive probability.
Let $E$ denote a positive probability event that some validator $i$ with cost function $c_i^*g_i(x)$ and capacity $\band_i$ gets assigned a transaction in BFM.
Such an event must exist by non-triviality of BFM and the choice of $(\boldsymbol{\band}^*,\boldsymbol f^*)$ and bid vector $\boldsymbol{b}$.
Since all validators' utility must be zero in the honest execution, we have
\[
  \mathbb E[r^*_i - c^*_ig_i(\beta^{\rm ass}_i)] = 0 \text{ for all } i\in[m],
\]
where $r^*_i$ is validator $i$'s reward in $\mathcal{M}$ and the expectation is taken over mechanism's randomness.
Moreover, $\beta^{\rm ass}_i > 0$ on event $E$.

Now consider a world that fixes other validators' true types $(\boldsymbol{\band}^*_{-i},\boldsymbol f^*_{-i})$, and validator $i$'s true type $(\band_i, \widetilde{c}_ig_i(x))$ for $0 < \widetilde{c}_i < c^*_i$.
In the honest case, validator $i$'s utility is zero.
However, if validator $i$ overreport its cost to be $c^*_i$ and then behave honestly otherwise,
it gets utility 
\[\mathbb E[r^*_i - \widetilde{c}_ig_i(\beta^{\rm ass}_i)] > \mathbb E[r^*_i - c^*_ig_i(\beta^{\rm ass}_i)] = 0,\]
where the strict inequality comes from the cost assumption that $g_i(x)>0$ for any $x>0$ and that $\beta^{\rm ass}_i > 0$ on the positive-probability event $E$.
This violates VIC, and we reach a contradiction.

\end{proof}

%% file: Sections/posted-price.tex
\section{Greedy-Route Posted-Price Mechanism}
\label{sec:greedy}

In this section, we introduce a mechanism based on posted-price, that achieves UIC and $1$-VIC. 
The mechanism specifies two parameters: the base fee $b > 0$ and reward parameter $0< \rev < b$. Here, $b$ is the price charged to each committed transaction, and $\rev$ is the per-unit reward for validator's verifiable work. These parameters are determined by the protocol.

\subsection{The Mechanism}
\label{sec:greedy-mechanism}
At a high level, the mechanism greedily assigns attestation routes based on a predetermined order.
Let $R^1,R^2,...R^L$ be an ordered list of $L = m\cdot \binom{m-1}{k}$ attestation routes, where each route $R^x = (s; A)$ specifies one source validator $s\in[m]$ and a size-$k$ set of attesters $A = \{a_1,...,a_k\}\subseteq [m]\setminus \{s\}$. The routes are ordered lexicographically, source first.
We describe the allocation as a scan over this list for clarity. However, \cref{sec:polynomial-route-scan} shows that it can be implemented exactly in $O(Nm^2k)$ time and $O(mk)$ working space, without enumerating the route list or requiring convexity of $g_i$.

Let $w_i(R)$ denote the allocated bandwidth usage of validator $i$ in a route $R=(s; A)$. Then
\begin{equation}
    w_i(R):=
    \begin{cases}
        k+1, &\text{ if } i = s,\\
        1, &\text{ if } i\in A,\\
        0, &\text{ otherwise.}
    \end{cases}
\end{equation}

Given a finalized ledger slice $\block = (\block_1,...,\block_m)$, define 
\[
B_i(\block):=\#\{\tx\in \block:{\sf src}(\tx)=i\},\qquad
A_i(\block):=\#\{\tx\in \block:i\in {\sf Att}(\tx)\},
\]
where ${\sf src}(\tx)$ is the source validator who proposed the blocks finalizing $\tx$, and ${\sf Att}(\tx)$ is the set of attesters providing signatures for $\tx$ in the finalized slice.

\paragraph{Priority-ordered transaction stream.}
We will let the planner to sort the received transaction based on some priority rule that cannot be manipulated by users.

For deterministic priority rules, $\prec$ is a fixed total order,  for example one given by planner-receipt timestamps.
For randomized priority rules, the planner chooses a random seed $\omega$, which then induces a total order $\prec_{\omega}$ over transaction identifiers. 
We use $\tx \prec^S_\omega \tx'$ to indicate that when sorting transactions in set $S$, possibly under random seed $\omega$,
$\tx$ is ranked earlier than $\tx'$.
We say that a priority rule is \emph{insertion-monotone} if for every transaction set $T$, every set $F$ of additional transactions such that $F\cap T = \emptyset$, then for every
$tx,tx'\in T$,
\[
tx \prec^T_\omega tx'
\quad\Longleftrightarrow\quad
tx \prec_\omega^{T\cup F} tx'.
\]
Equivalently, fake transactions may be inserted into the priority, but cannot change the relative order of the real transactions of a user.

\begin{mdframed}
        {\bf Greedy-Route Posted-Price} \\
        \noindent Parameters: base fee $b$ and reward parameter $0 < \rev < b$.
    \begin{itemize}
        \item {\bf Allocation} $\assign(\boldsymbol{f},\boldsymbol{\band}, \boldsymbol{b})$: Given reported cost functions $\boldsymbol{f} = ( c_i\cdot g_i(x))_{i \in [m]}$, bandwidth capacities $\boldsymbol{\band} = ( \band_i )_{i \in [m]}$ and bids $\boldsymbol{b}$, the planner first sorts transactions in $\boldsymbol{b}$ based on the priority rule $\prec$.
        For any validator with $c_i \geq \rev$, set $\band_i = 0$.
        We say that transactions which bid at least $b$ are \emph{eligible} transactions.
        Let $N$ be the total number of transactions.

        \underline{\it Greedy route.} 
        The central planner greedily scans the route list $R^1,...,R^L$ in order and accepts as many copies of each $R^x$ as possible subject to the bandwidth capacity and budget feasibility constraint. 
        
        \dotfill

        \begin{algorithmic}
    \State Initialize \(x \leftarrow 1\), \(\beta_i \leftarrow 0\) for every \(i \in [m]\),
    and \(\mathcal{R}(u) \leftarrow (\bot;\bot)\) for every \(u \in [N]\).

    \For{\(u \in [N]\)}
        \If{\(b_u < b\)}
            \State \textbf{continue}
        \EndIf

        \While{\(x \leq L\) and not both \eqref{eqn:cond1} and \eqref{eqn:cond2} hold for \(R^x\)}
            \State \(x \leftarrow x+1\)
        \EndWhile

        \If{\(x > L\)}
            \State \textbf{break}
        \EndIf

        \State Set \(\mathcal{R}(u) \leftarrow R^x = (s;A)\).
        \State For every \(i \in [m]\), set \(\beta_i \leftarrow \beta_i + w_i(R^x)\).
    \EndFor

    \State \textbf{return} \(\mathcal{R}\).
\end{algorithmic}
        The two conditions above are: 
        \begin{align}
            &\beta_i + w_i(R^x)\leq \band_i \text{ for all } i\in[m], \tag{cond1} \label{eqn:cond1}\\
            &\tau\cdot \sum_{i\in[m]} \left(g_i(\beta_i+w_i(R^x)) - g_i(\beta_i)\right)\leq b. \tag{cond2} \label{eqn:cond2}
        \end{align}

\dotfill

\noindent{\it Data dissemination.}
For every transaction $\tx$ with $\mathcal R(\tx)=(s(\tx);A(\tx))\ne(\bot,\bot)$,
the planner sends $\tx$ and the route $R(\tx)$ to source validator $s(\tx)$ with planner's signature. The source validator $s(\tx)$ sends $(\tx, R(\tx))$ to every attester $a\in A(\tx)$ and collects PoA.

An attester $a$ signs $\tx$ only if: 1.) $a\in A(\tx)$ in a planner-signed $R(\tx)$; 2.) $a$ has received the full transaction data; and 3.) the PoA header is consistent with the assigned source.

Let $\block = (\block_1,...,\block_m)$ be the finalized ledger slice, where $\block_i$ is proposed by validator $i$. 
\item {\bf Payment} $\price (\mathcal{R}, \block_1,...,\block_m)$: each committed transaction in $\block_1,...,\block_m$ pays the base fee $b$.
\item {\bf Reward} $\reward (\mathcal{R}, \block_1,...,\block_m)$: 
Let $B^{\rm ass}_i$ be the number of transactions in $\boldsymbol{b}$ such that $s(\tx) = i$.
Given the finalized ledger $\block$, the validator receives reward only if $B_i(\block) = B_i^{\rm ass}$.\footnote{As we will later show in \cref{clm:honest-dissemination}, we do not need to verify the number of signatures because the bandwidth cost of sending and receiving signatures is omitted.
If one wants to explicitly model the cost of sending/receiving signatures, one may add the check of number of attestations submitted by validator $i$.}
It receives 
\[r_i = \tau\cdot g_i(\beta_i(\block)),\]
where 
\[
\beta_i(\block):=(k+1)B_i(\block)+A_i(\block).
\]
\end{itemize}
\end{mdframed}

\paragraph{Streaming implementation.}
If $\prec$ is planner-receipt timestamp, then the allocation rule can be implemented online in a streaming manner.
The planner maintains the current route pointer $x$ and the current workload vector $(\beta_1,...,\beta_m)$. 
When a transaction arrives, the planner checks whether its bid is
at least $b$. If not, it assigns $(\bot,\bot)$. If yes, it assigns the transaction to the next available route in the greedy route scan, or assigns $(\bot,\bot)$ if no available route remains. 
Since later transactions cannot
improve the priority of earlier transactions and the greedy state only moves
forward, assignments made earlier in the round never need to be revised.

\subsection{Incentive Compatibility of Greedy-Route Posted-Price Mechanism}
\label{sec:greedy-feasibility-uic}

\begin{theorem}
\label{thm:main-greedy}
    For any integer $m\geq 3$, any $0 < \tau < b$,
    the above posted-price mechanism achieves UIC and $1$-VIC.
\end{theorem}
The attestation feasibility and individual rationality follows from the assign algorithm description. 
Budget feasibility follows from \eqref{eqn:cond2}.
Let \(\beta^\ell_i\) be validator \(i\)'s assigned workload after the first \(\ell\) accepted route copies. By \eqref{eqn:cond2},
\[
\tau\sum_i\bigl(g_i(\beta_i^\ell)-g_i(\beta_i^{\ell-1})\bigr)\le b
\]
for every accepted route \(\ell\). Summing over \(\ell=1,\ldots,\Theta\),
\[
\tau\sum_i g_i(\beta_i^\Theta)\le b\Theta.
\]
If rewards are paid on verifiable finalized work \(W_i(L)\le \beta_i^\Theta\), then
\[
\sum_i r_i(L)=\tau\sum_i g_i(W_i(L))
\le \tau\sum_i g_i(\beta_i^\Theta)
\le b\Theta,
\]
which is at most total base fee revenue from the \(\Theta\) charged transactions.

\begin{lemma}
The mechanism satisfies UIC.
\end{lemma}

\begin{proof}
Fix a user $u$ with true valuation $b^*_u$. We first argue that it is without loss to consider strategies in which $u$ submits only one bid for its real transaction.

Consider any deviation in which $u$ submits its real transaction together with an arbitrary non-negative number of eligible fake transactions (non-eligible transactions do not affect the allocation so we just ignore them).
By insertion-monotonicity of the priority rule, adding fake transactions cannot improve the priority position of $u$'s real transaction. 
Since the greedy assignment processes transactions according to this priority order, injecting fake transactions cannot increase the probability of the real transaction getting allocation.
Moreover, every finalized fake transaction needs to pay the base fee $b$.
Therefore, injecting fake transaction is not profitable.

It remains to analyze that any single bid is no better than honest bidding, which follows from truthfulness of standard posted-price auction. 
Therefore, the mechanism satisfies UIC.
\end{proof}

\begin{lemma}
\label{lem:vic}
    The mechanism satisfies $1$-VIC.
\end{lemma}

\begin{proof}
\newcommand{\honestsol}{\ensuremath{(\mathcal{R}^*,\Gamma^*, {\boldsymbol \theta}^*)}}
\newcommand{\sol}{\ensuremath{\mathcal{R}}}
Fix an arbitrary validator $i\in[m]$, its true type $(f_i^* = c_i^*\cdot g_i(x),\band_i^*)$, fix arbitrary other validators' true types $(f_j^*=c_j^*\cdot g_j(x),\band_j^*)$ for $j\neq i$ and the users' bid vector ${\boldsymbol b}$. 
We show that no strategy increases validator $i$'s utility compared to honest behavior, assuming other validators act honestly.
If validator $i$ crashes or outputs less than $B^{\rm ass}_i$ number of finalized transactions, its utility will be non-positive ($-\infty$ in the case of crashing, $-f_i^*(\beta^{\rm act}_i)$ in the case of outputting fewer than $B^{\rm ass}_i$ finalized transactions), which is no better than honest utility by individual rationality. Therefore, in the rest of this proof, we focus on \emph{admissible} strategies where validator $i$ does not crash and outputs $B^{\rm ass}_i$ number of finalized transactions.

A validator's strategy is a joint strategy consisting of several
components: (1) dissemination behavior, (2) fake-bid injection strategy, (3) reported cost, and (4) reported bandwidth capacity. These components may be chosen jointly and may interact.

The proof below uses a successive reduction argument. We do not analyze deviations on a single component while assuming that all other components are honest. 
Instead, at each step we fix the remaining components of the
strategy arbitrarily and show that the component under consideration can
be replaced by its honest version without increasing the
validator's utility. 
Thus that component is weakly dominated conditional on any choice of the other components. After applying these reductions one by one, every arbitrary joint strategy is transformed into the honest strategy. Hence no joint deviation can give higher utility than honest
behavior.

\paragraph{Misbehaving in data dissemination is not profitable.} 
First, we show that fix the validator's report and injected fake bids, misbehaving in data dissemination is not profitable.
\begin{claim}[Misbehaving in data dissemination is not profitable]
\label{clm:honest-dissemination}
Fix an arbitrary validator report $f_i$, $\band_i$ and fake bid $\widetilde{\bf b}$.
Then any strategy in which validator $i$ deviates in the data dissemination phase
is no better than behaving honestly in the data dissemination phase.
\end{claim}
\begin{proof}
For any fixed validator reports $(f_i, \band_i)$ and fake bids $\widetilde{b}$, the mechanism outputs a fixed assignment $\sol$. 
Denote the assigned transaction load as $B^{\rm ass}_i$ and the assigned attestation as $A^{\rm ass}_i = \sum_{(s;A)\in\mathcal{R}} \mathbf{1}_{\{i\in A\}}$.
The assigned bandwidth usage is thus $\beta^{\rm ass}_i=(k+1)B^{\rm ass}_i+A^{\rm ass}_i$.

Let $\block$ denote the finalized ledger slice under the validator $i$'s possibly randomized dissemination strategy. So $\block$ is a random variable. 
By our assumption, we have $B_i(\block) = B^{\rm ass}_i$.
Let $\beta^{\rm act}_i$ be the actual bandwidth usage and $A_i(\block)$ be the actual number of attestations $i$ provides in the finalized ledger slice. 
Then validator $i$'s utility is 
\[\rev\cdot g_i((k+1)\cdot B^{\rm ass}_i + A_i(\block))-f_i^*(\beta^{\rm act}_i).\]

Note that validator $i$ will receive $B^{\rm ass}_i$ transactions from the planner, disseminate $B^{\rm ass}_i$ transactions to get PoA, and received $A^{\rm ass}_i$ number of transactions to sign from other validators, so the actual bandwidth usage must satisfy
\begin{equation}
\label{eqn:no-better-dissemination}
    \beta^{\rm act}_i\geq (k+1)\cdot B^{\rm ass}_i+A^{\rm ass}_i = \beta^{\rm ass}_i.
\end{equation}
Moreover, since the validator cannot attest more than $A^{\rm ass}_i$ number of transactions as other validators are all behaving honestly, we have $A_i(\block)\leq A_i^{\rm ass}$.
Therefore,
\begin{equation}
\label{eqn:no-better-att}
    (k+1)\cdot B^{\rm ass}_i + A_i(\block)\leq \beta^{\rm ass}_i
\end{equation}
Putting together, the validator's expected strategic utility is
\[\mathbb E[\rev\cdot g_i((k+1)\cdot B^{\rm ass}_i + A_i(\block))-f_i^*(\beta^{\rm act}_i)] \leq \tau\cdot g_i(\beta^{\rm ass}_i) - f^*_i(\beta^{\rm ass}_i),\]
where the expectation is taken over randomness in validator strategy, and the inequality follows from \cref{eqn:no-better-dissemination} and \eqref{eqn:no-better-att}, and that $f^*_i(\cdot)$ is a strictly increasing function.

That is, a validator is incentivized to behave honestly based on the assigned workload during data dissemination. 
\end{proof}

\paragraph{Injecting fake transactions is not profitable.} 
By \cref{clm:honest-dissemination}, we only need to focus on strategies where the validators can misreport bandwidth costs and capacity, inject fake bids, but then behave honestly during data dissemination based on the assigned work.
Next, we show that the validator is not incentivized to inject any fake bids.
\begin{claim}[Injecting fake transactions is not profitable]
Fix an arbitrary validator report $f_i$ and $\band_i$ and assume that the validator behaves honestly during dissemination.
Then any strategy in which validator $i$ injects a positive number of fake transactions
is no better than injecting no fake transactions.
\end{claim}

\begin{proof}
Fake transactions with bids below $b$ are ineligible and do not affect the allocation. Thus we only need to consider eligible fake transactions. Fix the realization of any randomness used by the priority rule.

Let $H$ be the set of eligible non-fake transactions and let $F$ be the set of eligible fake transactions injected by validator $i$, with $|F|=n'$. 

Let $\mathcal R^H$ (resp. $\mathcal R^{H\cup F}$) denote the route-vector output when the submitted eligible transaction set is $H$ (resp. $H\cup F$), with the same validator reports and the same randomness (if any) used in the priority rule. Let $Q_H$ (resp. $Q_{H\cup F}$) be the number of transactions with non-$\bot$ attestation route in $\mathcal R^H$ (resp. $\mathcal R^{H\cup F}$).

Then 
\[
Q_{H\cup F}=Q_H+\Delta
\]
for some $0\le \Delta\le n'$.

We first bound the additional reward that validator $i$ can obtain. The only routes that can be newly assigned in the execution with fake transactions are the $\Delta$ additional suffix routes of the greedy route sequence. By the budget condition \eqref{eqn:cond2}, each such additional route increases total validator reward by at most $b$.
So validator $i$'s reward increases by at most $\Delta\cdot b$.
However, each assigned fake transaction will be finalized and pays the base fee. So the fake transaction payments are at least $\Delta\cdot b$.
Therefore, injecting fake transactions is not profitable.
\end{proof}

\paragraph{Misreporting $f_i$ is not profitable.} Therefore, now we only need to consider the strategies where the validator misreports cost function and bandwidth capacity but then inject no fake transactions and spends exactly $\beta^{\rm ass}_i$ bandwidth as assigned throughout the dissemination.

\begin{claim}[Misreporting $f_i$ is not profitable]
    For any report $f_i$ and $\band_i$, the validator's utility is no better than reporting $f^*_i$ and $\band_i$.
\end{claim}
\begin{proof}
    Note that $f_i$ is only used to determine whether validator $i$ should be assigned any workload.
    If $c^*_i \geq \rev$, for any report such that $c_i< \rev$, the validator gets a non-positive utility as it will spend a non-negative amount bandwidth at a per-unit marginal cost higher than $\rev$. Any report that $c_i\geq \rev$ results in the same utility as honestly reporting $c^*_i$ since the greedy allocation is the same.
    Therefore, the validator's utility is no better than reporting $f^*_i$ and $\band_i$.
    
    Similarly, if $c^*_i < \rev$, any report such that $c_i\geq \rev$ will result in zero utility as the mechanism will assign no work to validator $i$, which is no better than reporting $c^*_i$ honestly.
\end{proof}
\paragraph{Misreporting bandwidth is not profitable.} Now we are only left with strategies that misreports bandwidth capacity but otherwise behaves honestly. 
We prove the following useful claim on the monotonicity of the allocated bandwidth with respect to bandwidth capacity.

Let $\beta^{\rm ass}_i(\band_i)$ be the assigned bandwidth usage to validator $i$ in an execution where validator $i$ reports $\band_i$ but otherwise behave honestly.
\begin{claim}[Monotonicity of bandwidth-report]
\label{clm:band-mono}
If $\band'_i < \band_i$, then $\beta^{\rm ass}_i(\band_i') \leq \beta^{\rm ass}_i(\band_i)$.
\end{claim}

\begin{proof}
Fix all reports except validator $i$'s bandwidth report. Consider two executions of the greedy route algorithm: one where validator $i$ reports $\band_i$, and one with $\band_i' < \band_i$. All other validators' reported capacities are identical in the two executions.

If the route returned in both executions is the same, the claim directly follows. Therefore, suppose the two executions first differ when the algorithm considers some route $R$. Let $\beta_j$ denote the allocated bandwidth to validator $j$ for $j\in[m]$ at this moment.
Then the difference in the two allocations must come from validator $i$'s capacity.

Since the two executions first differ at $R$, it must be that $R$ is infeasible under the lower report $\band'_i$ (so $\beta_i + w_i(R) > \band'_i$) but feasible under the higher report $\band_i$ (so $\beta_i + w_i(R) \le \band_i$). Combining the two,
\[
\band'_i < \beta_i + w_i(R) \le \band_i.
\]

Since the assigned bandwidth usage only increases during the algorithm, it must be that
\[\beta^{\rm ass}_i(\band_i)\geq \overline{\beta}_i+w_i(R) >\band'_i \geq \beta^{\rm ass}_i(\band'_i),\]
where the last inequality follows from individual rationality of the mechanism.
\end{proof}

\begin{lemma}[Underreporting bandwidth is not profitable]
Underreporting bandwidth capacity $\band_i<\band_i^*$ is no better than truthful bandwidth reporting.
\end{lemma}

\begin{proof}
If validator $i$'s cost parameter $c^*_i \geq \tau$, then bandwidth capacity report does not change $i$'s allocation as the greedy assign algorithm will rewrite $\band_i = 0$.
In the rest of the proof, we focus on the case where $c^*_i <\tau$.
By \Cref{clm:band-mono}, we have $\beta^{\rm ass}_i(\band_i)\leq \beta^{\rm ass}_i(\band^*_i)$. Moreover, since the validator behaves honestly in the data dissemination, its utility $(\tau-c^*_i)g_i(\beta^{\rm ass}_i(\band^*_i)) \geq (\tau-c^*_i)g_i(\beta^{\rm ass}_i(\band_i))$.
Therefore underreporting bandwidth cannot increase validator $i$'s utility.
\end{proof}

\begin{lemma}[Overreporting bandwidth is not profitable]
Overreporting bandwidth capacity $\band_i > \band_i^*$ is no better than truthful bandwidth reporting.
\end{lemma}

\begin{proof}
We consider two cases.

\noindent
\textbf{Case 1: $\beta^{\rm ass}_i(\band_i)>\band_i^*$.}
Then validator $i$ is assigned more bandwidth than its true capacity. If it attempts to execute the assigned work, it either crashes and receives $-\infty$ utility, or fails to finish the assigned work and receives zero reward.
Therefore overreporting is not profitable in this case.

\medskip
\noindent
\textbf{Case 2: $\beta^{\rm ass}_i(\band_i)\leq \band_i^*$.}
Subject to not exceeding the bandwidth capacity, the overreport execution produces the same accepted route list as truthful reporting.

Suppose, for the sake of contradiction, that the two executions first differ when the greedy algorithm considers adding a copy of some route $R$.
Let $\beta_i$ be the current assigned bandwidth usage of validator $i$ at this moment.

Since the overreport gives validator $i$ more reported capacity, it must be the case that the overreporting execution accepts $R$ while the truthful execution rejects $R$.
Since all other validators have the same capacities in both executions, it must be the case that the difference comes from validator $i$'s report, i.e.,
\[\beta_i+w_i(R)>\band_i^*.\]
However, if the overreport execution accepts $R$, it must assign more workload than the real capacity, i.e., $\beta^{\rm ass}_i(\band_i)\geq \beta_i+w_i(R) > \band_i^*$. This contradicts the assumption of Case 2.

Combining the two cases, overreporting bandwidth cannot increase validator $i$'s utility.
\end{proof}
\end{proof}

\subsection{Approximation in the Capacity-Limited Regime}
\label{sec:greedy-approximation}

We compare the greedy-route allocation to the best allocation that could be obtained by an omniscient planner using the same base fee $b$ and validator reward rate $\tau$.
Fix $\tau<b$ and overload the notation to represent the effective capacity
\[
    \band_i(\tau)
    :=
    \begin{cases}
        \band_i, & c_i\leq \tau,\\
        0, & c_i>\tau.
    \end{cases}
\]
as in the greedy-route case.
This effective bandwidth encodes validator individual rationality for private costs $f^*_i(x)=c^*_i \cdot g_i(x)$: a validator with $c_i>\tau$ cannot be profitably assigned a positive amount of work at reward rate $\tau$.  
We assume there are infinitely many eligible transactions, so the approximation factor is the factor comparing the system's throughput.
Given the reward parameter $\tau$, 
the optimal allocation benchmark is the following integer program.  The variable $y_x$ is the number of transactions routed along route $R^x$, and $\beta^{\rm ass}_i$ is the assigned bandwidth usage of validator $i$, given the validator reports $(\boldsymbol{\band},\boldsymbol{f})$ and a bid vector of infinitely many eligible bids.
\begin{equation}
\label{eqn:opt-bfm-ip}
\begin{aligned}
\mathrm{OPT}_{\tau}
:=\max\quad & \sum_{x=1}^{L} y_x\\
\text{s.t. }
& \beta^{\rm ass}_i=\sum_{x=1}^{L} y_x\cdot w_i(R^x) && \forall i\in[m],\\
& 0\leq \beta^{\rm ass}_i\leq \band_i(\tau) && \forall i\in[m],\\
& \tau\cdot \sum_{i\in[m]} g_i(\beta^{\rm ass}_i)\leq b\cdot \sum_{x=1}^{L} y_x,\\
& y_x\in\mathbb Z_{\geq 0} && \forall x\in[L].
\end{aligned}
\end{equation}
The capacity constraints ensure that no validator is assigned more work than its reported bandwidth allows.  The budget constraint requires total validator reward to be covered by the base-fee revenue from the committed transactions.

\begin{theorem}[Greedy approximation in the capacity-limited regime]
\label{thm:approx}
Fix posted-price parameters \(b,\tau\), the number of validators \(m\ge 3\), and
the PoA threshold \(k<m\). Let \(G_\tau\) be the number of transactions committed
by the greedy-route BFM on infinitely many eligible bids, and let
\(\beta^G_i\) be the final greedy workload of validator \(i\).

Assume the final greedy allocation is capacity-maximal: for every route
\(R=(s;A)\), there exists some validator \(i\in \{s\}\cup A\) such that
\[
    \beta^G_i+w_i(R)>\operatorname{cap}_i(\tau).
\]
Then
\[
    G_\tau
    \ge
    \frac{1}{2k+1}OPT_\tau
    -
    \frac{(k+1)m}{2k+1}.
\]
In particular, when \(OPT_\tau\gg m\), the greedy-route BFM is a
\((2k+1)\)-approximation to the optimal feasible allocation.
\end{theorem}

\begin{proof}
Let \(\beta^G_i\) be validator \(i\)'s final workload under the greedy allocation.
Each greedy transaction consumes $q := 2k+1$, so
\[
    \sum_{i=1}^m \beta^G_i = q\cdot G_\tau.
\]

Consider an arbitrary feasible optimal allocation achieving \(OPT_\tau\).
For each transaction in this optimal allocation, let \(R=(s;A)\) be its route.
By the capacity-maximality assumption, this route is blocked in the final greedy
allocation by at least one validator \(i\in\{s\}\cup A\), meaning
\[
    \beta^G_i+w_i(R)>\operatorname{cap}_i(\tau).
\]
Charge this optimal transaction to one such blocking validator.

We now bound the number of optimal transactions charged to any fixed validator
\(i\). Every transaction charged to \(i\) uses validator \(i\) in its optimal
route, and therefore consumes at least one unit of \(i\)'s bandwidth in the
optimal allocation. Since the optimal allocation is capacity feasible, the
number of transactions charged to \(i\) is at most \(\operatorname{cap}_i(\tau)\).

Moreover, whenever \(i\) blocks a route \(R\), we have
\[
    \operatorname{cap}_i(\tau) < \beta^G_i+w_i(R)
    \le \beta^G_i+k+1,
\]
because \(w_i(R)\le k+1\) for every route. Hence the number of optimal
transactions charged to validator \(i\) is at most \(\beta^G_i+k+1\).

Summing over all validators gives
\[
    OPT_\tau
    \le
    \sum_{i=1}^m(\beta^G_i+k+1)
    =
    qG_\tau+(k+1)m.
\]
The theorem follows by rearranging the above inequality.
\end{proof}

\Cref{thm:approx} focuses the approximation analysis on the capacity-congestion regime. The benchmark in \eqref{eqn:opt-bfm-ip} remains the correct optimal comparison even when the budget constraint binds, but the simple greedy-route analysis no longer has the blocking structure used above: a route rejected for budget reasons need not be blocked by any validator’s residual capacity. 
In the fully cost-congested regime, a fixed route order is not enough to guarantee a meaningful approximation independent of the cost profiles. 
An unlucky ordering can allocate early transactions to routes whose marginal rewards are high, therefore making later routes fail the marginal budget check, even though a different ordering of the routes might accept more transactions if some lower-cost validators appeared first. Obtaining an approximation guarantee in this regime will require designing a cost-aware route ordering that incentivizes the validators to report truthfully.

\begin{remark}
    We assumed for simplicity that every transaction consumes one unit of bandwidth. But our greedy-route mechanism works for arbitrary size transactions as follows.
    One can scale the workload of $w_i(R)$ and the base fee $b$ based on its size when deciding whether this transaction can be served.
    The mechanism's incentive properties are unaffected. This is because the UIC and $1$-VIC proofs rely only on the workload increments being nonnegative and additive, and on $g_i$ being strictly increasing, not on the increments being of unit size. However, allowing arbitrary-size does influence the approximation guarantee of Theorem~\ref{thm:approx}, where the constant $2k+1$ would instead depend on the size of the largest transaction.
\end{remark}

\subsection{Polynomial-Time Implementation of the Greedy Route Scan}
\label{sec:polynomial-route-scan}

The allocation rule can be implemented without enumerating the $L = m\binom{m-1}{k}$ routes. We keep the transaction priority, the eligibility condition, and the forward-only route pointer of the greedy-route mechanism, and as before, set $\mathrm{cap}_i = 0$ whenever $c_i \ge \tau$.

\paragraph{Route order.}
We write every attester set in increasing order, $A = \{a_1 < \cdots < a_k\}$, and represent a route $(s;A)$ by its source and this sorted list. The route order used by the greedy scan is source-major and lexicographic in the attesters: $(s;A)$ precedes $(s';A')$ if and only if $s < s'$, or $s = s'$ and $(a_1,\dots,a_k)$ lexicographically precedes $(a'_1,\dots,a'_k)$. Specifically, $R^1 = (1;\{2,\dots,k+1\})$. 

\paragraph{Marginal rewards for a fixed source.}
Fix the current workloads $(\beta_1,\dots,\beta_m)$ and a candidate source $s$. If $\beta_s + k + 1 > \mathrm{cap}_s$, then no route with source $s$ satisfies~\eqref{eqn:cond1}, and $s$ is skipped. Otherwise, let
\[
  \Delta^{\mathrm{src}}_s := \tau\bigl(g_s(\beta_s + k + 1) - g_s(\beta_s)\bigr),
\]
and, for each $j \in [m]$,
\[
  \Delta_j :=
  \begin{cases}
    \tau\bigl(g_j(\beta_j + 1) - g_j(\beta_j)\bigr), & j \ne s \text{ and } \beta_j + 1 \le \mathrm{cap}_j,\\[2pt]
    +\infty, & \text{otherwise.}
  \end{cases}
\]
Since $w_i(s;A) = 0$ for every $i \notin \{s\} \cup A$, the sum in~\eqref{eqn:cond2} reduces to the marginal rewards of the source and the attesters. So a size-$k$ set $A \subseteq [m]$ yields a route $(s;A)$ satisfying both~\eqref{eqn:cond1} and~\eqref{eqn:cond2} if and only if $\beta_s + k + 1 \le \mathrm{cap}_s$ and 

\begin{equation}
\label{eqn:route-feasible}
  \Delta^{\mathrm{src}}_s + \sum_{j \in A} \Delta_j \le b .
\end{equation}
The workloads, and therefore all marginal rewards, are held fixed during the search for a single route.

\paragraph{Dynamic program over attesters.}
For $j \in \{1,\dots,m+1\}$ and $r \in \{0,\dots,k\}$, let $D[j,r]$ be the minimum of $\sum_{\ell \in S}\Delta_\ell$ over all $S \subseteq \{j,\dots,m\}$ with $|S| = r$, with $D[j,r] = +\infty$ if every such sum is infinite or no such $S$ exists. Set $D[j,0] = 0$ for $1 \le j \le m+1$ and $D[m+1,r] = +\infty$ for $1 \le r \le k$, and for $j = m, m - 1, \dots, 1$ and
$r = 1, \dots, k$ compute

\begin{equation}
\label{eqn:dp}
  D[j,r] = \min\bigl\{\, D[j+1,r],\ \Delta_j + D[j+1,r-1] \,\bigr\},
\end{equation}
where the two terms correspond to skipping and selecting validator $j$. The table has $O(mk)$ entries and is filled in $O(mk)$ time. Its indices are validator positions and attester counts. This means there is no enumeration of routes and no discretization of capacities or budgets.

\paragraph{Lexicographically first completion.}
Given $j \in [m+1]$, $r \in \{0,\dots,k\}$, and a remaining budget $B_{\mathrm{rem}}$ with $D[j,r] \le B_{\mathrm{rem}}$, the procedure $\textsc{Complete}(j, r, B_{\mathrm{rem}})$ scans $\ell = j, j+1, \dots$ in increasing order and selects $\ell$ whenever

\begin{equation}
\label{eqn:complete-test}
  \Delta_\ell + D[\ell+1, r-1] \le B_{\mathrm{rem}},
\end{equation}
in which case it sets $B_{\mathrm{rem}} \gets B_{\mathrm{rem}} - \Delta_\ell$ and $r \gets r - 1$. Otherwise it skips $\ell$. It stops when $r = 0$. Test~\eqref{eqn:complete-test} asks whether selecting $\ell$ still leaves a feasible completion, rather than just asking whether $\ell$ fits the budget. The procedure runs in $O(m)$ time.

\paragraph{Successor search from the route pointer.}
Let the current pointer be $R^x = (s_0; \{a_1 < \cdots < a_k\})$. The procedure $\textsc{Next}(R^x)$ scans sources $s = s_0, s_0 + 1, \dots, m$, skipping any $s$ with $\beta_s + k + 1 > \mathrm{cap}_s$, and for each remaining $s$ computes $\Delta^{\mathrm{src}}_s$, the marginal rewards $(\Delta_j)_{j\in[m]}$, and the table $D$.
\begin{itemize}
  \item For $s > s_0$: if $\Delta^{\mathrm{src}}_s + D[1,k] \le b$, return $(s; \textsc{Complete}(1, k, b - \Delta^{\mathrm{src}}_s))$. Otherwise continue with the next source.
  \item For $s = s_0$: if $\{a_1,\dots,a_k\}$ satisfies~\eqref{eqn:route-feasible}, return $R^x$. Otherwise compute the
  prefix sums $P_0 = 0$ and $P_t = \sum_{\ell=1}^{t} \Delta_{a_\ell}$ for $t \in [k]$, and examine pairs $(t, j)$ in the order $t = k, k-1, \dots, 1$ and for each $t$, $j = a_t + 1, \dots, m$. For the first pair satisfying
  \begin{equation}
  \label{eqn:prefix-test}
    \Delta^{\mathrm{src}}_{s_0} + P_{t-1} + \Delta_j + D[j+1, k-t] \le b,
  \end{equation}
  return the route with source $s_0$ and attesters $\{a_1,\dots,a_{t-1}, j\} \cup \textsc{Complete}\bigl(j+1,\, k-t,\, b - \Delta^{\mathrm{src}}_{s_0} - P_{t-1} - \Delta_j\bigr)$. If no pair passes, continue with source $s_0 + 1$.
\end{itemize}
If no source yields a route, $\textsc{Next}(R^x)$ returns $\bot$. The search never enumerates the $\binom{m-1}{k}$ attester sets of any source.

\paragraph{Full allocation rule.}
Initialize $\beta_i \gets 0$ for all $i$, $R(u) \gets (\bot;\bot)$ for all $u$, and the pointer to $R^1$. Process transactions in priority order, skipping those with $b_u < b$. For each eligible transaction $u$, call $\textsc{Next}$ on the current pointer and workloads. If it returns $\bot$, terminate the scan, leaving all remaining transactions unassigned. Otherwise set $R(u)$ to the returned route, move the pointer to that same route, and update $\beta_i \gets \beta_i + w_i(R(u))$ for all $i$. The pointer is not advanced after an acceptance, so the next transaction may reuse the same route if it remains feasible. All marginal rewards and tables are recomputed at the updated workloads.

\begin{lemma}
\label{lem:complete-correct}
Fix the marginal rewards and suppose $D[j,r] \le B_{\mathrm{rem}}$. Then $\textsc{Complete}(j,r,B_{\mathrm{rem}})$ returns the lexicographically first set $S \subseteq \{j,\dots,m\}$ with $|S| = r$ and $\sum_{\ell\in S}\Delta_\ell \le B_{\mathrm{rem}}$.
\end{lemma}

\begin{proof}
By induction on $j$ from $m+1$ downward, recurrence~\eqref{eqn:dp} gives the minimum, since every $r$-subset of $\{j,\dots,m\}$ either omits or contains $j$. We now induct on $r$. The case $r = 0$ is trivial. A feasible $r$-subset exists because $D[j,r] \le B_{\mathrm{rem}}$. Let $\ell^*$ be the smallest $\ell \ge j$ passing~\eqref{eqn:complete-test}. If $S$ is any feasible subset with smallest element $\ell''$, then $\Delta_{\ell''} + D[\ell''+1, r-1] \le \sum_{\ell \in S}\Delta_\ell \le B_{\mathrm{rem}}$, so $\ell''$ passes the test and $\ell'' \ge \ell^*$. Conversely, $\ell^*$ extends to a feasible subset by the definition of $D$. Since the lexicographic order compares smallest elements first, the lexicographically first feasible subset has smallest element $\ell^*$, and its remainder is the lexicographically first feasible $(r-1)$-subset of $\{\ell^*+1,\dots,m\}$ under budget $B_{\mathrm{rem}} - \Delta_{\ell^*}$. The scan skips exactly the elements $j,\dots,\ell^*-1$, selects $\ell^*$, and continues from $\ell^* + 1$ with this budget, which satisfies the hypothesis of the lemma, so the induction hypothesis applies.
\end{proof}

\begin{lemma}
\label{lem:next-correct}
For any workloads $(\beta_i)_{i\in[m]}$ where $0\leq \beta_i\leq {\sf cap}_i$ for $i\in[m]$ and pointer $R^x$, $\textsc{Next}(R^x)$ returns $R^y$ for the smallest $y \ge x$ such that $R^y$ satisfies~\eqref{eqn:cond1} and~\eqref{eqn:cond2}, and returns $\bot$ if no such $y \le L$ exists.
\end{lemma}

\begin{proof}
In the route order, the routes at or after $R^x = (s_0;A)$ are: $R^x$, then the routes $(s_0;A')$ with $A'$ lexicographically after $A$, then all routes with source $s > s_0$ in increasing order of $s$. Sources without residual capacity have no route satisfying~\eqref{eqn:cond1}, so skipping them is safe. For $s > s_0$, \Cref{lem:complete-correct} with $j = 1$ and $r = k$ returns the first feasible route with source $s$ whenever one exists, which happens exactly when $\Delta^{\mathrm{src}}_s + D[1,k] \le b$.

For $s = s_0$, every $A' = \{a'_1 < \cdots < a'_k\}$ lexicographically after $A$ has a unique first position $t$ at which it differs from $A$, and $a'_t = j > a_t$. Call $(t,j)$ the group of $A'$. The group consists of exactly the sets $\{a_1,\dots,a_{t-1}, j\} \cup T$ with $T \subseteq \{j+1,\dots,m\}$ and $|T| = k-t$. If $t > t'$, every set in group $(t,j)$ precedes every set in group $(t',j')$: both agree with $A$ before position $t'$, and at position $t'$ the former has $a_{t'}$ while the latter has $j' > a_{t'}$. If $t = t'$ and $j < j'$, the former again precedes the latter. Within one group, the order is the lexicographic order of $T$. By the definition of $D$, group $(t,j)$ contains a set satisfying~\eqref{eqn:route-feasible} if and only if~\eqref{eqn:prefix-test} holds. So the first feasible set after $A$ lies in the first group passing~\eqref{eqn:prefix-test} in the examined order, and by \Cref{lem:complete-correct} it is the set returned.
\end{proof}

\begin{theorem}
\label{thm:polytime-greedy}
The implementation above produces the same assignment $R$ and the same final workloads as the greedy-route allocation rule. For $N$ transactions already in priority order, it runs in $O(Nm^2k)$ time with $O(mk)$ working space, assuming constant-time evaluation of each $g_i$ and constant-time arithmetic and comparisons.
\end{theorem}

\begin{proof}
We show by induction over processed transactions that the pointer and the workloads coincide with those of the original algorithm. Both algorithms skip ineligible transactions. For an eligible transaction, the original while-loop advances the pointer from $R^x$ to the smallest $y \ge x$ for which $R^y$ satisfies~\eqref{eqn:cond1} and~\eqref{eqn:cond2} under the current workloads, or terminates if none exists. By \Cref{lem:next-correct}, $\textsc{Next}$ returns the same route or $\bot$. Both algorithms then assign that route, keep the pointer at it, and apply the same workload update, or both terminate.

For the running time, each source examined by $\textsc{Next}$ costs $O(m)$ to compute the marginal rewards, $O(mk)$ to fill $D$, $O(k)$ for the prefix sums, $O(km)$ for the prefix tests at $s_0$, and $O(m)$ for the completion, for $O(mk)$ in total. A call examines at most $m$ sources, giving $O(m^2 k)$ per eligible transaction. The table is reused across sources, so the working space is $O(mk)$ beyond the inputs and the stored assignment. Sorting into priority order adds $O(N \log N)$ when needed. Under planner-receipt order, the implementation works in a streaming fashion without the need for sorting.
\end{proof}

Note that no convexity assumption on the public shape functions $g_i$ is required. If some $g_i$ is concave, a route rejected for budget reasons may become feasible again once workloads grow. The greedy-route mechanism never revisits such a route because its pointer only moves forward, and the implementation reproduces exactly this forward-only behavior.

%% file: Sections/multi-blocks.tex
\section{Multi-Block Price Dynamics}
\label{sec:multi_block_dynamics}

In this section, we present a posted price update rule of our greedy-route BFM and study the local behavior of the price under honest validator behavior. For this section, we assume that the validators report their true bandwidth capacity and cost once in round $t = 0$. 
Then, in the following rounds, they follow the assignment and implement the dissemination honestly. 
Throughout, we assume the reward parameter $\tau$ remains the same.
We focus on the capacity-limited regime in which the route-level budget constraint is slack in a neighborhood of the operating point.
Therefore, the maximum available bandwidth $\Theta^{\max}$ is determined based on the validators' types and $\tau$.
Let $q:=2k+1$ be the total bandwidth consumed by one successfully disseminated transaction.  If $\mathcal R^\infty$ denotes the multiset of accepted routes returned by $\assign(\boldsymbol{f},\boldsymbol{\band},\boldsymbol{b})$ when $\boldsymbol{b}$ contains infinitely many eligible transactions, then the maximum transaction throughput is $N^{\max}:=|\mathcal R^\infty|$ and the maximum aggregate bandwidth is
\[
    \Theta^{\max}:=qN^{\max}=q|\mathcal R^\infty|.
\]
Thus $\Theta^{\max}$ is measured in bandwidth units, not in transactions.

We adopt a price-update rationale similar to that of EIP-1559. Specifically, in each round, there is a target bandwidth usage. 
The price decreases if usage in the previous round was below the target, to admit more transactions, and increases otherwise.
However, in our multi-block setting, the available supply changes over time: validators consume bandwidth when transactions are disseminated and gradually recover bandwidth between rounds. Thus the posted price and the available bandwidth evolve jointly.

We will define the exact stochastic process and introduce a deterministic fluid approximation. 
Under mild conditions on $\lambda$ and the step size $d$ (made precise in Theorems~\ref{thm:fluid_fixed_point} and~\ref{thm:fluid_local_stability}), the fluid system has a unique, locally stable, positive-price fixed point, which implies that, over any fixed finite horizon and from nearby initial states, the normalized stochastic process tracks the fluid dynamics in large markets with high probability.

\subsection{The Stochastic Supply-Price Process}
Let $\Theta^{\max}$ denote the system's maximum aggregate bandwidth and let $q = 2k+1$ denote the amount of bandwidth consumed by each committed transaction: $k+1$ units from the source validator and one unit from each of the $k$ attesters.
Let $\Theta_t\in[0,\Theta^{\max}]$ denote the available bandwidth at the beginning of round $t$, and let $B_t>0$ denote the base fee in round $t$.
For simplicity, we allow $\Theta_t$ to be real-valued for the dynamic process analysis.
Define the normalized available supply and stochastic state by
\[
    X_t:=\frac{\Theta_t}{\Theta^{\max}},
    \qquad
    Z_t:=(X_t,B_t).
\]
\paragraph{Assumptions.} We will make the following assumptions on users.
\begin{assumption}[Arrival Process]
\label{ass:arrivals}
In each round $t$, the number of arriving users satisfies \[N_t\sim \mathrm{Poisson}(\lambda\Theta^{\max}),\]
where $\lambda>0$ is the arrival rate relative to the maximum supply. The random variables $N_t$ are independent across rounds.
\end{assumption}

\begin{assumption}[Valuation Distribution]
\label{ass:valuations}
Each arriving user $i$ in round $t$ has valuation $v_i^t$, drawn independently from a common continuous distribution with CDF $F$ whose support is a subset of $\mathbb{R}_{\geq 0}$. 
We assume that $F(0) = 0$.

A user is eligible at base fee $B_t$ if $v_i^t\ge B_t$. Let \[
    S_t := \sum_{i=1}^{N_t}\mathbf 1_{\{v_i^t\ge B_t\}}
\]
be the number of eligible users in round $t$, where $\mathbf 1_{\{v_i^t\ge B_t\}}$ is the indicator of whether user $i$ is eligible.
Then $S_t\sim \mathrm{Poisson}\left(\lambda\Theta^{\max}(1-F(B_t))\right)$.
\end{assumption}

\paragraph{Bandwidth dynamics.} In each round, the system consumes bandwidth to disseminate transactions while the previously used bandwidth slowly recovers with a recovery rate $\alpha$.
Given state $(\Theta_t,B_t)$, the real-valued analytical process consumes bandwidth $\min\{\Theta_t,qS_t\}$ in round $t$.
The random normalized consumption is
\begin{equation}
\label{eqn:norm-g}
    G_t
    :=
    \min\left\{
    X_t,
    \frac{qS_t}{\Theta^{\max}}
    \right\}.
\end{equation}
This real-valued relaxation is convenient for analysis.  The transaction-level simulations in \Cref{sec:simulation} use the indivisible refinement $q\min\{\lfloor\Theta_t/q\rfloor,S_t\}$, whose normalized one-step discrepancy is less than $q/\Theta^{\max}$.
Fix a recovery rate $\alpha\in(0,1]$. 
The bandwidth then evolves according to $\Theta_{t+1}=\Theta_t-G_t\cdot \Theta^{\max}+\alpha(\Theta^{\max}-\Theta_t)$.
Equivalently,
\begin{equation}
\label{eqn:normalized-band}
    X_{t+1}=X_t-G_t+\alpha(1-X_t).
\end{equation}

\paragraph{Base-fee update.}

The base fee is updated in the style of EIP-1559. We use half of the currently available supply $\Theta_t/2$ as a per-round target:
\[
  B_{t+1} =
  \begin{cases}
    \displaystyle
    B_t\!\left(1 + d\cdot\frac{G_t - X_t/2}{X_t/2}\right) & \text{if } X_t > 0,\\[6pt]
    B_t & \text{if } X_t = 0,
  \end{cases}
\]
where $d\in(0,1)$ is the step size. Because $d<1$, the multiplicative factor $1 + d \cdot (\cdots)$ satisfies $1-d \le 1+d\cdot \frac{G_t-X_t/2}{X_t/2}\le 1+d$, so $B_{t+1}>0$ whenever $B_t>0$. When more than half of the available supply is consumed, the base fee increases. When less than half is consumed, it decreases. When $X_t = 0$, no transactions are disseminated, so $G_t = 0$ and the ratio $\frac{G_t}{\frac{X_t}{2}}$ is undefined, so we keep the price unchanged.

\begin{lemma}[Markov property]
\label{lem:markov}
The process $\{Z_t\}_{t\ge 0}$ is a time-homogeneous Markov process.
\end{lemma}

\begin{proof}
Let $Z_0,\ldots,Z_t$ be the history up to round $t$, and suppose $Z_t=(u,p)$. By Assumption ~\ref{ass:arrivals}, the number of arrivals $N_t$ is independent across rounds. By Assumption ~\ref{ass:valuations}, conditional on $Z_t=(u,p)$, the distribution of $S_t$ depends only on $p$. The consumed bandwidth $G_t$ depends only on the current state $Z_t$ and fresh round-$t$ randomness $S_t$. Once $G_t$ is realized, both $X_{t+1}$ and $B_{t+1}$ are deterministic functions of $(X_t,B_t,G_t)$, with $B_{t+1}=B_t$ when $X_t=0$. So for every measurable set $A\subseteq[0,1] \times \mathbb{R}_{>0}$:
\[
  \Pr[Z_{t+1}\in A\mid Z_0,\ldots,Z_t] = \Pr[Z_{t+1}\in A\mid Z_t].
\]
The transition rule and the distribution of fresh randomness are identical every round, so the process is time-homogeneous.
\end{proof}

\subsection{Fluid Approximation}

We consider the following deterministic fluid approximation. Let
\[
    g(u,p):=\min\{u,q\lambda(1-F(p))\}
\]
denote the normalized fluid consumption at state $(u,p)$.
Let $z_t:=(x_t,b_t)$ denote the fluid state. Its evolution is given by
\begin{equation}
\label{eq:fluid_system}
\begin{aligned}
    x_{t+1}
    &=
    x_t-g(x_t,b_t)+\alpha(1-x_t),\\
    b_{t+1}
    &=
    \begin{cases}
    b_t\left(
        1+d\cdot
        \dfrac{g(x_t,b_t)-x_t/2}{x_t/2}
    \right),
    & x_t>0,\\[2ex]
    b_t,
    & x_t=0.
    \end{cases}
\end{aligned}
\end{equation}

Equivalently, define the fluid update map $H: [0,1]\times\mathbb R_{>0} \mapsto [0,1]\times\mathbb R_{>0}$ as
\begin{equation}
\label{eqn:fluid-update}
    H(u,p):=
    \begin{cases}
    \left(u-g(u,p)+\alpha(1-u),\;p\left(1+d\cdot\dfrac{g(u,p)-u/2}{u/2}\right)\right),
    & u>0,\\
    (\alpha,p),
    & u=0.
    \end{cases}
\end{equation}
Then $z_{t+1}=H(z_t)$.
We stress that $\mathbb E[G_t\mid Z_t=(u,p)]$
is not exactly equal to $g(u,p)$, because of the minimum truncation in the definition of $G_t$.
Instead, the fluid approximation replaces the random number of eligible users $S_t$ by its conditional mean $\mu(p):=\lambda\Theta^{\max}(1-F(p))$
before applying the capacity constraint. 

The following lemma quantifies the resulting one-step approximation error.

\begin{lemma}[One-step fluid approximation error]
\label{lem:fluid_error_onestep}
For every state $(u,p)\in[0,1]\times\mathbb R_{>0}$,
\begin{equation}
\label{eq:one_step_expectation_error}
    0
    \le
    g(u,p)-\mathbb E[G_t\mid Z_t=(u,p)]
    \le
    q\sqrt{\lambda/\Theta^{\max}}.
\end{equation}
Moreover, for every $\varepsilon>0$,
\begin{equation}
\label{eq:one_step_tail_bound}
\Pr\left[|G_t-g(u,p)|\ge\varepsilon\mid Z_t=(u,p)\right]\le
    2\exp\left(-\frac{\Theta^{\max}\varepsilon^2}{2q(q\lambda+\varepsilon)}\right).
\end{equation}
\end{lemma}

\begin{proof}
When $u=0$, both $G_t$ and $g(u,p)$ are zero, so both claims
hold trivially. In the rest of the proof we focus on the case where $u > 0$.
By Jensen's inequality 
\begin{align*}
    \mathbb E[G_t\mid Z_t=(u,p)]
    &=
    \mathbb E\left[
        \min\left\{
            u,\frac{qS_t}{\Theta^{\max}}
        \right\}
        \,\middle|\,
        Z_t=(u,p)
    \right]\\
    &\le
    \min\left\{
        u,
        \frac{
            q\mathbb E[S_t\mid Z_t=(u,p)]
        }{
            \Theta^{\max}
        }
    \right\}=
    g(u,p).
\end{align*}
Moreover, we have
\[|G_t-g(u,p)|=
    \left|\min\left\{u,\frac{qS_t}{\Theta^{\max}}\right\}-\min\left\{
            u,\frac{q\mu(p)}{\Theta^{\max}}\right\}\right|\le
    \frac{q}{\Theta^{\max}}|S_t-\mu(p)|.\]
Taking conditional expectations and applying Cauchy--Schwarz,
\begin{align*}
    g(u,p)-\mathbb E[G_t\mid Z_t=(u,p)]
    &\le
    \frac{q}{\Theta^{\max}}
    \mathbb E\left[
        |S_t-\mu(p)|
        \,\middle|\,
        Z_t=(u,p)
    \right]\\
    &\le
    \frac{q}{\Theta^{\max}}
    \sqrt{
        \operatorname{Var}
        (S_t\mid Z_t=(u,p))
    }\\
    &\le
    q\sqrt{\lambda/\Theta^{\max}},
\end{align*}
which proves \eqref{eq:one_step_expectation_error}.

For the tail bound, if $S\sim\operatorname{Poisson}(\mu)$, then for every
$r>0$, we have $\Pr[|S-\mu|\ge r]\le2\exp\left(-\frac{r^2}{2(\mu+r)}\right).$
Taking $r = \frac{\Theta^{\max}\varepsilon}{q}$, we get
\begin{align*}
    &\Pr\left[
        |G_t-g(u,p)|\ge\varepsilon
        \,\middle|\,
        Z_t=(u,p)
    \right]\\
    &\quad\le
    2\exp\left(
        -
        \frac{
            (\Theta^{\max})^2\varepsilon^2/q^2
        }{
            2\left(
                \mu(p)+\Theta^{\max}\varepsilon/q
            \right)
        }
    \right)\\
    &\quad\le
    2\exp\left(
        -
        \frac{
            \Theta^{\max}\varepsilon^2
        }{
            2q(q\lambda+\varepsilon)
        }
    \right),
\end{align*}
where the last inequality comes from $1-F(p)\le 1$.
\end{proof}

The bound in~\eqref{eq:one_step_expectation_error} is one-sided. $g(u,p)$ is always a (weak) overestimate of the true expected consumption. This follows directly from Jensen's inequality applied to the concave minimum truncation $s \mapsto \min\{u, qs/\Theta^{\max}\}$.

\begin{claim}
\label{lem:fluid_map_lipschitz}
Suppose $F$ is Lipschitz on a compact interval $I=[p_{\min},p_{\max}]
    \subset\mathbb R_{>0}$.
Then the fluid update map $H$ is Lipschitz on every compact rectangle $K=[u_{\min},u_{\max}]\times I\subset(0,1]\times\mathbb R_{>0}$.
\end{claim}

\begin{proof}
Let $L_F$ be a Lipschitz constant for $F$ on $I$. For any $(u,p),(u',p')\in K$,
\[
  \vert g(u,p) - g(u',p')\vert = \bigl \vert \min\{u, q\lambda(1-F(p))\} - \min\{u',q\lambda(1-F(p'))\}\bigr \vert \le \vert u - u' \vert + q \lambda L_F \vert p - p'\vert.
\]
Thus $g$ is Lipschitz on $K$. Therefore the supply component $H_x(u,p) = u - g(u,p) + \alpha(1-u)$ satisfies:
\[
  \vert H_x(u,p) - H_x(u',p')\vert \le \vert u-u'\vert + \vert g(u,p) - g(u',p') \vert + \alpha \vert u-u'\vert \le C_1 \bigl(\vert u-u'\vert + \vert p-p'\vert \bigr)
\]
for some constant $C_1$ depending only on $\alpha$ and $L_F$.

For the price component, let $R(u,p) := 1 - d + \frac{2dg(u,p)}{u}$, which is well-defined for $u>0$. Since $0 \le g(u,p) \le u$:
\[
  1-d \le R(u,p)\le 1+d.
\]
Because $u,u' \ge u_{\min} > 0$, we have:
\[
  \left \vert \frac{g(u,p)}{u} - \frac{g(u',p')}{u'}\right \vert \le \frac{\vert g(u,p)-g(u',p')\vert}{u_{\min}} +\frac{\vert u-u' \vert }{u_{\min}},
\]
so $R$ is Lipschitz on $K$. Since $p,p'\le p_{\max}$ on $K$,
\begin{align*}
  \vert H_y(u,p) - H_y(u',p') \vert &= |p R(u,p) - p' R(u',p')\vert\\
  &\le (1+d) \vert p-p'\vert + p_{\max} \vert R(u,p) - R(u',p')\vert
\end{align*}
which is Lipschitz. So $H$ is Lipschitz on $K$.
\end{proof}

We next show that the stochastic process is a small-noise perturbation of the fluid system over every fixed finite horizon, provided that the fluid trajectory stays away from the boundary $x=0$.

\begin{lemma}[Finite-horizon tracking]
\label{lem:tracking}
Fix a finite horizon $T \in \mathbb{N}$, which is held constant as $\Theta^{\max} \to \infty$, and suppose $F$ is locally Lipschitz. Let $z_t =(x_t,b_t)$ be a fluid sequence with $x_t>0$ for all $t=0,\ldots,T$. Starting from the same state $Z_0 = z_0$, we have:
\[
  \max_{0 \le t \le T}\|Z_t - z_t\| = O_p\left(\frac{1}{\sqrt{\Theta^{\max}}}\right),
\]
and $\max_{0 \le t \le T}\|Z_t-z_t\| \xrightarrow{p}0$ as $\Theta^{\max} \to \infty$.
\end{lemma}

\begin{proof}
There exists an $\varepsilon>0$ and a compact set $K\subset(0,1]\times\mathbb R_{>0}$ such that the closed $\varepsilon$-neighborhood of $\{z_0,...,z_T\}$ is contained in $K$.
By \Cref{lem:fluid_map_lipschitz}, $H$ is Lipschitz on $K$ with some Lipschitz constant $L_K$.

Define the one-step consumption error by $\delta_t:= G_t-g(X_t,B_t).$
Whenever $Z_t=(X_t,B_t)\in K$, the stochastic update can be written as $Z_{t+1}=H(Z_t)+\Delta_t$,
where
\[
    \Delta_t
    =
    \left(
        -\delta_t,\;
        \frac{2dB_t}{X_t}\delta_t
    \right).
\]
Since $X_t$ is bounded away from zero and $B_t$ is bounded above on $K$, whenever $Z_t \in K$, there exists a finite constant $C_K$ such that:
\[
    \|\Delta_t\|\le C_K|\delta_t|
\]

Let $\sigma:=\min\{t\le T:Z_t\notin K\}$ be the first time that $Z_t$ does not fall inside $K$, or $\sigma = T+1$ if the process never leaves $K$ by time $T$.
Then for every $t<\sigma$,
\[
    \|Z_{t+1}-z_{t+1}\|=\|H(Z_t)+\Delta_t-H(z_t)\|\le
    L_K\|Z_t-z_t\|+C_K|\delta_t|.
\]
Starting from $Z_0 = z_0$, we get for every $t\leq \sigma$,
\begin{equation}
\label{eqn:horizon-error}
    \|Z_t-z_t\| \le A_t \max_{0\le s<t}|\delta_s|,
\end{equation}
where $A_t = C_K \left( \sum_{j=0}^{t-1}L_K^j \right)$.

Consider the event $E_{\varepsilon}$ that $A_T\underset{0\le s<T}{\max}|\delta_s|<\varepsilon$.
\begin{claim}
    Conditioned on event $E_{\varepsilon}$, we have $\sigma = T+1$, i.e., the process does not leave $K$ before time $T$. 
\end{claim}
\begin{proof}
    Suppose for the sake of contradiction that $\sigma \leq T$. By definition of $\sigma$, we have 
    \[ \|Z_\sigma - z_\sigma\| \le A_T\max_{0\le s<T}|\delta_s| < \varepsilon.\]
    Thus $Z_\sigma$ belongs to the closed $\varepsilon$-neighborhood of $z_\sigma$, which is contained in $K$. This contradicts the definition of $\sigma$.
\end{proof}
By the tail bound in Lemma~\ref{lem:fluid_error_onestep} and a union bound over $t = 0,...,T$, we have 
\[\max_{0\le s<T}|\delta_s| = O_p\left(
        \frac{1}{\sqrt{\Theta^{\max}}}
    \right).\]
Since $A_T$ is a constant independent of $\Theta^{\max}$, it follows that 
$\Pr[E_{\varepsilon}]\rightarrow 1$, and therefore, $\Pr[\sigma = T+1]\rightarrow 1$.
In this case, \Cref{eqn:horizon-error} holds for any $t\leq T$.
The lemma thus follows as $A_T$ is a constant independent of $\Theta^{\rm max}$.
\end{proof}

\subsection{Fixed Point and Local Stability of the Fluid Approximation}
For our fluid system~\eqref{eqn:fluid-update}, a fixed point $(x^*,b^*)\in(0,1)\times\mathbb{R}_{>0}$ satisfies $x_{t+1}=x^*$ and $b_{t+1}=b^*$ whenever $x_t=x^*$ and $b_t=b^*$.

\begin{theorem}[Unique fixed point]
\label{thm:fluid_fixed_point}
If $\lambda > \frac{\alpha}{q(2\alpha+1)}$,
the fluid system \eqref{eq:fluid_system} has a positive-price fixed point $(x^*,b^*)$ satisfying
\begin{equation}
    \label{eq:fixed-point}
    x^*=\frac{2\alpha}{2\alpha+1}, \qquad F(b^*) = 1-\frac{\alpha}{q\lambda(2\alpha+1)}.
\end{equation}
If $F$ is strictly increasing, this positive-price fixed point is unique.
\end{theorem}

\begin{proof}
At a positive-price fixed point with $x^*>0$, the price update leaves the price unchanged only if the multiplicative update factor is one, i.e., $g(x^*,b^*)=x^*/2$.
The supply update leaves supply unchanged only if $x^* = x^*-g(x^*,b^*)+\alpha(1-x^*)$. Combining both conditions, we have $x^* = 2\alpha/(2\alpha+1)$ and $g(x^*,b^*) = x^*/2$.

Since capacity constraint is slack at the fixed point, it must be that $g(x^*,b^*)=\frac{\alpha}{2\alpha+1}= q\lambda(1-F(b^*))$.
Equivalently,
\[
    F(b^*)
    =
    1-\frac{\alpha}{q\lambda(2\alpha+1)}.
\]
By $\lambda > \frac{\alpha}{q(2\alpha+1)}$, we have $1-\frac{\alpha}{q\lambda(2\alpha+1)}\in (0,1)$.
Since $F(0) = 0$ and $F$ is continuous, such a $b^*$ exists. If $F$ is strictly increasing, this $b^*$ is unique. Since $x^*$ is also uniquely determined, the positive-price fixed point is unique.
\end{proof}

\begin{theorem}[Exact local stability boundary]
\label{thm:fluid_local_stability}
Assume the conditions of \Cref{thm:fluid_fixed_point}. Suppose moreover that $F$ is continuously differentiable in a neighborhood of $b^*$, with density $f(b^*)>0$. Define relative-deviation coordinates
\[
    u_t:=\frac{x_t-x^*}{x^*},
    \qquad
    v_t:=\frac{b_t-b^*}{b^*},
\]
and let
\[
    \phi:=q\lambda f(b^*),
    \qquad
    \kappa:=\frac{\phi b^*}{x^*}
    =
    \frac{q\lambda b^* f(b^*)}{x^*}.
\]
In a neighborhood where the capacity constraint is slack, the Jacobian in $(u_t,v_t)$ coordinates is
\[
    J(d)
    =
    \begin{pmatrix}
        1-\alpha & \kappa\\
        -d & 1-2d\kappa
    \end{pmatrix}.
\]
Let
\[
    d_{\mathrm{crit}}
    =
    \frac{1}{\kappa}
    \begin{cases}
    \dfrac{4-2\alpha}{3-2\alpha},
    &0<\alpha\leq \frac12,\\[1.2ex]
    \min\left\{
    \dfrac{4-2\alpha}{3-2\alpha},
    \dfrac{\alpha}{2\alpha-1}
    \right\},
    &\frac12<\alpha\leq1.
    \end{cases}
\]
Then the positive-price fixed point $(x^*,b^*)$ is locally asymptotically stable for $0 < d < \min \{1, d_{\mathrm{crit}}\}$. At $d = d_{\mathrm{crit}}$, the linearization has spectral radius exactly one. For $d_{\mathrm{crit}} < d < 1$, the fixed point is locally linearly unstable, subject to the same differentiability and non-binding-capacity neighborhood assumptions.
\end{theorem}

\begin{proof}
Since $g(x^*,b^*)=x^*/2<x^*$, the capacity constraint remains slack in a sufficiently small neighborhood of the fixed point.  In that neighborhood,
\[
    \frac{\partial g}{\partial x}(x^*,b^*)=0,
    \qquad
    \frac{\partial g}{\partial b}(x^*,b^*)=-q\lambda f(b^*)=-\phi.
\]
Linearizing the normalized fluid map in the relative-deviation coordinates $(u_t,v_t)$ gives the displayed matrix $J(d)$.  Its trace and determinant are
\[
    \operatorname{tr}J=2-\alpha-2d\kappa,
    \qquad
    \det J=(1-\alpha)+d\kappa(2\alpha-1).
\]
For a two-dimensional discrete-time linear system with characteristic polynomial $z^2-(\operatorname{tr}J)z+\det J$, Schur stability is equivalent to the three Jury conditions
\[
    1-\operatorname{tr}J+\det J>0,
    \qquad
    1+\operatorname{tr}J+\det J>0,
    \qquad
    1-\det J>0.
\]
Substitution gives:
\begin{align*}
  1 - \operatorname{tr}J + \det J = (1+2\alpha)d\kappa > 0,
\end{align*}
so the first Jury condition imposes no upper bound for $d > 0$. The second condition gives
\[
  1 + \operatorname{tr}J + \det J = 4 - 2\alpha - d\kappa(3-2\alpha) > 0,
\]
The third condition gives
\[
  1-\det J = \alpha - d\kappa(2\alpha-1) > 0.
\]
When $0< \alpha \le \frac{1}{2}$ this is automatically satisfied for $d>0$. When $\alpha > \frac{1}{2}$ it gives the additional bound $d\kappa < \frac{\alpha}{2\alpha-1}$. Combining the binding upper bounds yields $d_{\mathrm{crit}}$. The standard linearization theorem gives local asymptotic stability for $d<d_{\mathrm{crit}}$, spectral radius one at the boundary, and local linear instability beyond the relevant boundary.
\end{proof}

\subsection{Implications for the Stochastic Process}

The local stability result of Theorem~\ref{thm:fluid_local_stability} applies to the deterministic fluid approximation.
We next discuss its implication for the actual stochastic process $Z_t$, which is a small-noise perturbation of the fluid system.

\begin{lemma}[Stochastic implication of fluid stability]
\label{cor:stochastic_fluid_stability}
Assume the conditions of Theorem~\ref{thm:fluid_local_stability}. Let $U$ be a sufficiently small neighborhood of the fluid fixed point $(x^*,b^*)$, and fix a finite horizon $T$. Then there exists a smaller neighborhood $V\subseteq U$ such that, if $Z_0=z_0\in V$, then
\[
    \Pr\left[
    Z_t\in U
    \text{ for all }0\le t\le T
    \right]
    \to 1
\]
as $\Theta^{\max}\to\infty$.
\end{lemma}

\begin{proof}
First we will choose an open set $W$ with $(x^*,b^*) \in W$ and $\overline{W} \subset U$. Define
\[
  \eta := \operatorname{dist} \bigl(\overline{W}, U^c \bigr) > 0.
\]
By Theorem~\ref{thm:fluid_local_stability}, $(x^*,b^*)$ is locally asymptotically stable for the fluid map $H$. So there exists an open neighborhood $V \subseteq W$ of $(x^*, b^*)$ such that every fluid trajectory $z_0 \in V$ satisfies $z_t\in W$ for all $t = 0, \ldots, T$. Now fix $z_0 \in V$. The fluid sequence $z_0, z_1, \ldots, z_T$ lies in $W$ and is therefore bounded away from the boundary $x = 0$. By Lemma~\ref{lem:tracking} (applied with the current fixed $T$) we have:
\begin{equation}
    \label{eq:stochastic-limit}
    \max_{0\le t\le T}\|Z_t-z_t\| = O_p\left(\tfrac{1}{\sqrt{\Theta^{\max}}}\right)
    \xrightarrow{p} 0 \quad\text{as }\Theta^{\max}\to\infty.
\end{equation}

Let $\mathcal{E}:=\bigl\{\max_{0\le t\le T}\|Z_t-z_t\|<\eta\bigr\}$. On $\mathcal{E}$, for each $t\le T$, since $z_t\in\overline{W}$,
\[
  \operatorname{dist}(Z_t,U^c) \ge\operatorname{dist}(z_t,U^c)-\|Z_t-z_t\| > \eta - \eta=0,
\]
so $Z_t\in U$. Since $\Pr[\mathcal{E}] \to 1$ by~\eqref{eq:stochastic-limit}, the claim follows.

Note that for any fixed starting point $z_0 \in V$ and fixed horizon $T$, the stochastic process stays near the fluid fixed point with high
probability when the market is large. The lemma does not reflect on long-run behavior beyond this fixed horizon $T$.

\end{proof}

The fluid fixed point should be interpreted as the center of the large-market stochastic dynamics. When $\Theta^{\max}$ is large, the normalized stochastic process tracks the stable fluid system over finite horizons, and near the locally stable fixed point its typical one-step normalized fluctuations are of order $1/\sqrt{\Theta^{\max}}$.

%% file: simulation.tex
\section{Numerical Evaluation}
\label{sec:simulation}
In this section, we perform some numerical simulations for the dynamic bandwidth-pricing model. The experiments are organized around three questions raised by the analysis in Sections~\ref{sec:impossibility}--~\ref{sec:multi_block_dynamics}. First, does the indivisible, transaction-level stochastic system track the real-valued fluid approximation once the initial transient is removed, and does the fixed point correctly predict steady-state levels as the recovery rate and market size vary (Sections~\ref{subsec:sim-baseline} and~\ref{subsec:sim-sensitivity}). Second, is the local stability threshold identified in Theorem~\ref{thm:fluid_local_stability} visible as an operational threshold in simulated price trajectories, and not just as an asymptotic guarantee (Section~\ref{subsec:sim-stepsize}). Third, does the discretization introduced by indivisible transactions and finite capacity produce a detectable gap relative to the real-valued model, and under what conditions does that gap become relevant (Section~\ref{subsec:sim-jensen}). Section~\ref{subsec:sim-route} ends with a concrete greedy-route instance that connects the aggregate bandwidth value $\Theta^{\max}$ used throughout the simulations back to the route-level mechanism of Section~\ref{sec:greedy}.

\subsection{Experimental setup}
\label{subsec:sim-setup}

Recall that $q:=2k+1$ is the total bandwidth consumed by one successfully disseminated transaction and that $\Theta^{\max}=q|\mathcal R^\infty|$ is measured in bandwidth units. Thus $\Theta_t$ is measured in bandwidth units, while $X_t=\Theta_t/\Theta^{\max}$ is dimensionless. The arrival parameter $\lambda$ is kept in the normalization of \Cref{ass:arrivals}: at posted price $b$, the eligible-arrival mean is $\lambda\Theta^{\max}(1-F(b))$. In the baseline experiments, $k=2$, so $q=5$, and $\Theta^{\max}=10000$ bandwidth units.

The theory in \Cref{sec:multi_block_dynamics} allows $\Theta_t$ to be real-valued. The experiments simulate the natural indivisible refinement in which transactions are accepted in integer units. If $S_t$ denotes eligible arrivals, the real-valued normalized load in \Cref{eqn:norm-g} is
\[
    G_t
    =
    \min\left\{
        X_t,
        \frac{qS_t}{\Theta^{\max}}
    \right\},
\]
whereas the simulation uses transaction capacity
\[
    C_t=\left\lfloor \frac{\Theta_t}{q}\right\rfloor
\]
and normalized realized load
\[
    \widetilde G_t
    =
    \frac{q}{\Theta^{\max}}\min\{C_t,S_t\}.
\]
Equivalently, the simulated unnormalized bandwidth load is
\[
    g_t=q\min\left\{\left\lfloor \frac{\Theta_t}{q}\right\rfloor,S_t\right\}.
\]
The deterministic discretization discrepancy is bounded by one transaction:
\[
    0
    \le
    \min\{\Theta_t,qS_t\}
    -
    q\min\left\{
        \left\lfloor\Theta_t/q\right\rfloor,
        S_t
    \right\}
    <q,
\]
so the normalized one-step discrepancy is less than $q/\Theta^{\max}$. This error vanishes in the large-market normalization used by the finite-horizon fluid approximation; we do not use this observation as a stronger stochastic convergence theorem.

Unless otherwise stated, the experiments use $F=\mathrm{Uniform}[1,2]$, $\lambda=0.15$, $\alpha=0.25$, $d=0.05$, $T=20000$ rounds, and a burn-in of $5000$ rounds. Aggregate sweeps use $50$ independent seeds after burn-in, except the step-size stress test uses $30$ seeds and the exact conditional-load diagnostic uses $5$ seeds. All trajectories start from $\Theta_0=\Theta^{\max}$ and $b_0=0.72b^\star$. For the uniform baseline, \Cref{thm:fluid_fixed_point} gives
\[
    X^\star=1/3,
    \qquad
    \Theta^\star\approx3333.3,
    \qquad
    g^\star\approx1666.7,
    \qquad
    b^\star\approx1.778.
\]
Reported RMSE values are normalized by the corresponding fixed-point value. Confidence intervals are normal-approximation $95\%$ intervals across independent seeds.

\subsection{Baseline tracking and fixed-point behavior}
\label{subsec:sim-baseline}

\Cref{fig:sim-baseline} plots three stochastic trajectories together with the deterministic real-valued fluid trajectory from the same initial state, zoomed after the short initial transient. The stochastic paths fluctuate around the fixed point, while the fluid path approaches it smoothly. This supports interpreting the fixed point as the operating point of the stochastic system rather than as an absorbing state. Together with the discretization bound in Section~\ref{subsec:sim-setup}, this baseline check confirms that the indivisible simulation is an accurate implementation of the fluid model at the parameter scale used throughout the remaining experiments, rather than an independent source of drift that we would need to control for separately.

\begin{figure}[H]
    \centering
    \includegraphics[width=.86\linewidth]{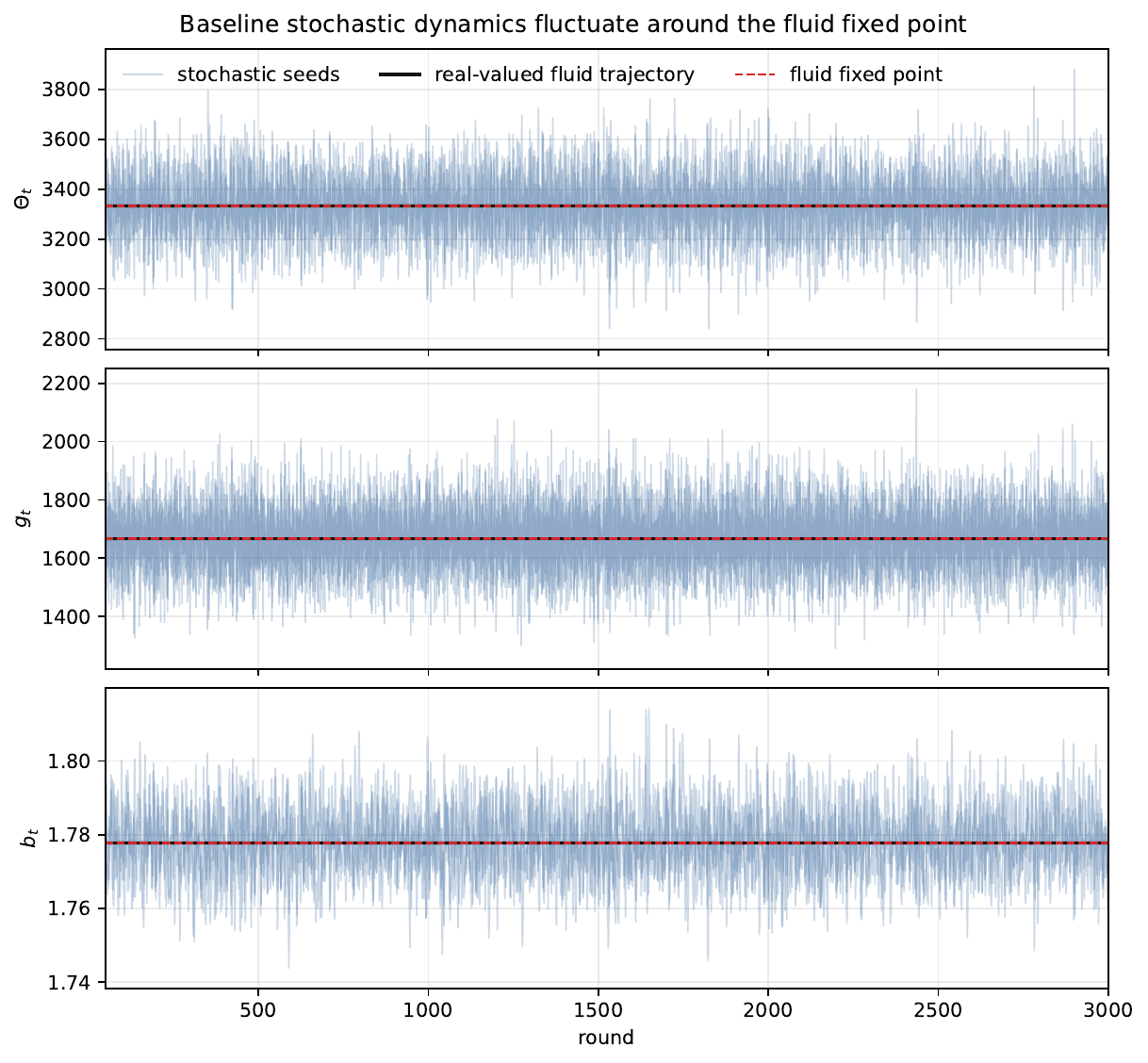}
    \caption{Baseline dynamics after the short initial transient. The plotted stochastic trajectories use indivisible transactions; the black curve is the real-valued fluid trajectory.}
    \label{fig:sim-baseline}
\end{figure}

\subsection{Recovery-rate and market-size sensitivity}
\label{subsec:sim-sensitivity}

The fixed point predicts that long-run available bandwidth increases with the recovery rate as
\[
    \Theta^\star(\alpha)=\frac{2\alpha}{1+2\alpha}\Theta^{\max}.
\]
\Cref{tab:sim-recovery,fig:sim-recovery} compare this prediction with post-burn-in averages for five recovery rates. The simulated averages closely track the closed-form predictions. Faster recovery supports higher steady-state throughput, while the target-load rule keeps consumption near one half of current available bandwidth.

\begin{table}[H]
    \centering
    \input{tables/table_recovery_sweep.tex}
    \caption{Recovery-rate sweep. Means are computed over $50$ seeds after burn-in and are reported in bandwidth units for $\Theta$ and $g$.}
    \label{tab:sim-recovery}
\end{table}

\begin{figure}[H]
    \centering
    \includegraphics[width=.78\linewidth]{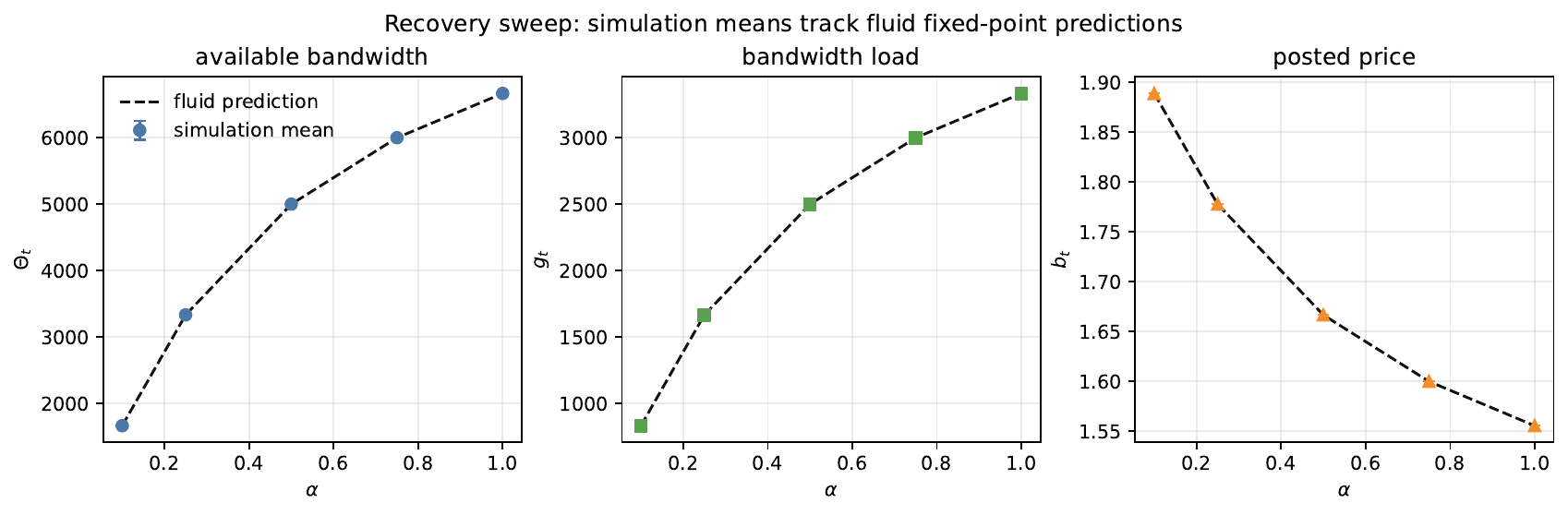}
    \caption{Recovery sweep in absolute units. Dashed curves show the fluid fixed-point predictions and markers show post-burn-in simulation means with $95\%$ confidence intervals.}
    \label{fig:sim-recovery}
\end{figure}

Having confirmed that the fixed point correctly predicts steady-state levels across recovery rates, we next ask how quickly the stochastic system concentrates around this fixed point as the market grows. \Cref{fig:sim-market-size} varies $\Theta^{\max}\in\{500,1000,5000,10000,50000\}$ while holding $k,\lambda,\alpha,d$ fixed. The normalized post-burn-in errors shrink as market size grows, which is empirically consistent with the finite-horizon large-market scaling in \Cref{lem:tracking}; it is not a separate proof of the limiting result.

\begin{figure}[H]
    \centering
    \includegraphics[width=.78\linewidth]{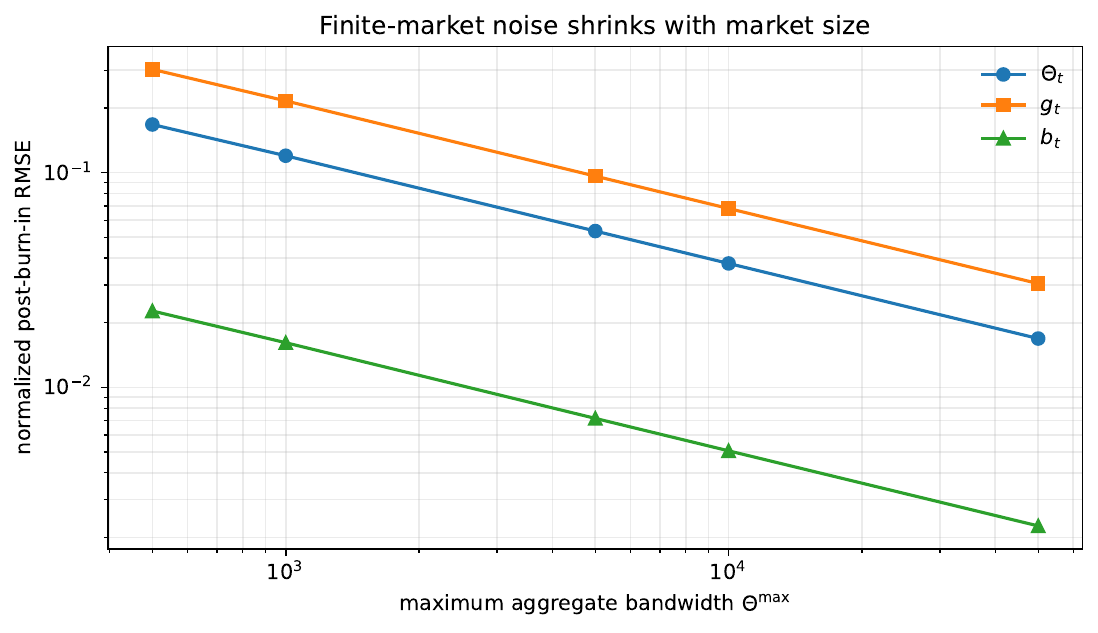}
    \caption{Finite-market concentration. Normalized stochastic deviations shrink as the maximum aggregate bandwidth $\Theta^{\max}$ increases.}
    \label{fig:sim-market-size}
\end{figure}

\subsection{Step size and local stability}
\label{subsec:sim-stepsize}

The simulations use the exact local Schur-stability boundary in \Cref{thm:fluid_local_stability}. For the uniform baseline,
\[
    \kappa=\frac{q\lambda b^\star f(b^\star)}{x^\star}=4,
    \qquad
    d_{\mathrm{crit}}=0.35.
\]
\Cref{fig:sim-stepsize,tab:sim-stepsize} vary $d$ below, at, and above this threshold. The RMSE statistics use the full post-burn-in horizon specified in the setup. The trace panel plots a downsampled $5000$-round representative path for readability.

\begin{figure}[H]
    \centering
    \includegraphics[width=.98\linewidth]{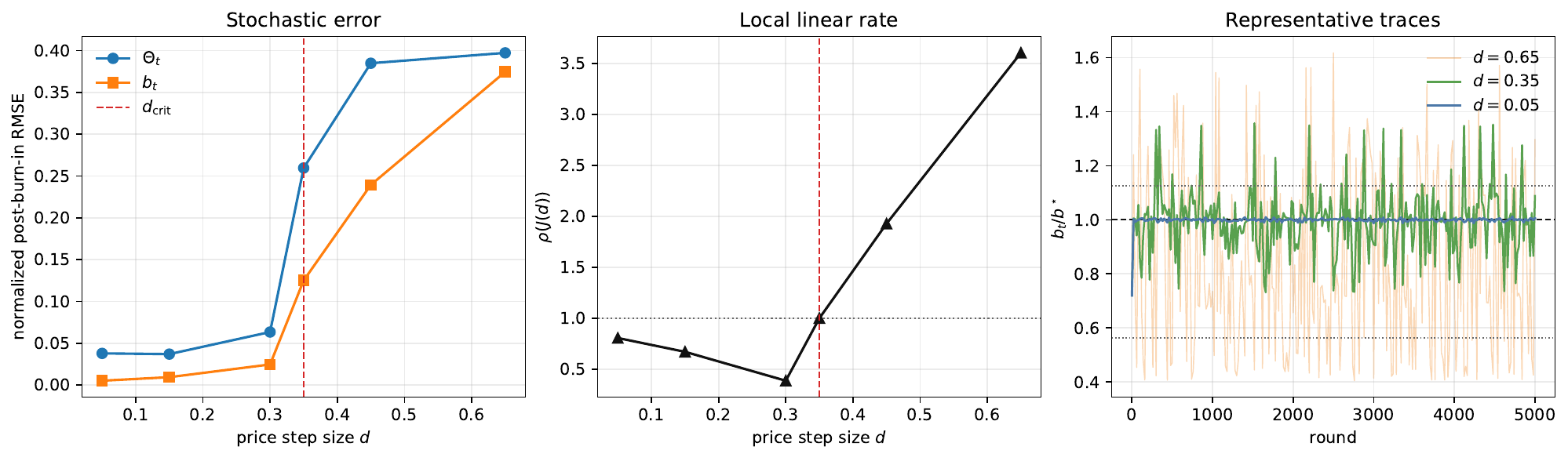}
    \caption{Step-size stress test. Left: post-burn-in stochastic RMSE over the full simulation horizon. Middle: deterministic spectral radius of the local Jacobian. Right: downsampled representative price traces over $5000$ plotted rounds.}
    \label{fig:sim-stepsize}
\end{figure}

\begin{table}[H]
    \centering
    \input{tables/table_step_size_stability.tex}
    \caption{Step-size diagnostics. The boundary row has $\rho(J(d))=1$ up to numerical precision.}
    \label{tab:sim-stepsize}
\end{table}

Table~\ref{tab:sim-stepsize} makes the content of Theorem~\ref{thm:fluid_local_stability} visible at the level of sample statistics. For $d$ below $d_{\mathrm{crit}}$, the spectral radius $\rho(J(d))$ stays below one and the normalized RMSE for both $\Theta_t$ and $b_t$ remains flat and small. At $d=d_{\mathrm{crit}}$ the boundary row already shows a visible jump in RMSE despite $\rho(J(d))=1$ only up to numerical precision. For $d>d_{\mathrm{crit}}$ the local linearization is unstable and the RMSE roughly doubles to triples, depending on the series, relative to the boundary case. The representative traces in the right panel of Figure~\ref{fig:sim-stepsize} show this transition directly. The $d=0.05$ and $d=0.35$ paths stay close to $b^\star$, while the $d=0.65$ path shows the oscillations characteristic of a locally unstable linearization.

The experiments so far treat $\Theta_t$ as real-valued, consistent with the fluid model of Section~\ref{sec:multi_block_dynamics}. We now examine the indivisibility assumption introduced in Section~\ref{subsec:sim-setup}, checking whether the resulting discretization produces a measurable gap between the fluid load and the exact conditional mean load of the discrete simulation, and identifying the conditions where this gap becomes non-negligible.

\subsection{Capacity binding and Jensen gap}
\label{subsec:sim-jensen}

For transaction capacity $C=\lfloor\Theta/q\rfloor$ and $A\sim\mathrm{Poisson}(\mu)$, the exact conditional mean load of the discrete simulation is
\[
    \overline g(\Theta,b)
    =
    q\,\mathbb E[\min\{C,A\}],
    \qquad
    \mu=\lambda\Theta^{\max}(1-F(b)).
\]
The real-valued fluid load is $\widehat g(\Theta,b)=\min\{\Theta,q\mu\}$. The comparison is
\[
    q\,\mathbb E[\min\{C,A\}]
    \leq
    q\min\{C,\mathbb E[A]\}
    \leq
    \min\{\Theta,q\mathbb E[A]\}.
\]
The first inequality follows from concavity and Jensen's inequality; the second follows from $qC\leq\Theta$.

\Cref{tab:sim-diagnostics} reports an implementation sanity check for the exact conditional mean. In the baseline regime the fixed-point cap is about $18.2$ Poisson standard deviations above the mean, so capacity binding is negligible and the Jensen gap is numerically zero. A deliberately small tight-cap case with $\Theta^{\max}=100$ makes truncation visible while keeping the same valuation and update-rule parameters.

\begin{table}[H]
    \centering
    \input{tables/table_diagnostics.tex}
    \caption{Diagnostics for cap binding, the Jensen gap, and the concrete route instance. The tight-cap column is a stress case, not the baseline regime.}
    \label{tab:sim-diagnostics}
\end{table}

\subsection{Concrete route instance}
\label{subsec:sim-route}

Finally, we instantiate a concrete linear-cost route instance in which the budget constraint is slack. There are $m=20$ validators and $k=2$, hence $q=5$. The capacity vector is $(497,540,\ldots,540)$. We use
\[
    g_i(x)=x,
    \qquad
    f_i(x)=c_i x,
\]
with $\tau=1$, $c_i=0.5$ for every validator, and route budget $b=5$. Since $q\tau=5\leq b$ and $c_i\leq\tau$ for every used validator, every capacity-feasible route is also budget-feasible. The capacity-only greedy scan therefore matches the full greedy mechanism on this instance.

Scanning routes lexicographically and accepting as many copies as possible yields $|\mathcal R^\infty|=2000$ accepted transactions. Because each transaction consumes $q=5$ bandwidth units,
\[
    q|\mathcal R^\infty|=5\cdot2000=10000=\Theta^{\max}.
\]
This check is only an aggregate-realizability witness for the baseline bandwidth value; it is not a universal validation of all greedy-route instances or a load-balancing benchmark.

%% file: tables/table_recovery_sweep.tex
\begin{tabular}{rrrrrr}
\toprule
$\alpha$ & $\Theta^\star$ pred. & $\Theta$ mean & $g^\star$ pred. & $g$ mean & $b$ mean \\
\midrule
0.10 & 1666.7 & 1665.6 & 833.3 & 833.4 & 1.889 \\
0.25 & 3333.3 & 3334.7 & 1666.7 & 1666.3 & 1.778 \\
0.50 & 5000.0 & 5001.2 & 2500.0 & 2499.4 & 1.667 \\
0.75 & 6000.0 & 6000.9 & 3000.0 & 2999.3 & 1.600 \\
1.00 & 6666.7 & 6667.4 & 3333.3 & 3332.6 & 1.556 \\
\bottomrule
\end{tabular}

%% file: tables/table_step_size_stability.tex
\begin{tabular}{rrrrr}
\toprule
$d$ & $\rho(J(d))$ & Stable? & $\Theta$ RMSE & $b$ RMSE \\
\midrule
0.05 & 0.806 & yes & 0.038 & 0.005 \\
0.15 & 0.671 & yes & 0.037 & 0.009 \\
0.30 & 0.387 & yes & 0.063 & 0.025 \\
0.35 & 1.000 & boundary & 0.259 & 0.125 \\
0.45 & 1.928 & no & 0.385 & 0.239 \\
0.65 & 3.603 & no & 0.397 & 0.374 \\
\bottomrule
\end{tabular}

%% file: tables/table_diagnostics.tex
\begin{tabular}{lrr}
\toprule
Diagnostic & Baseline & Tight cap \\
\midrule
$\Theta^{\max}$ & 10000 & 100 \\
Cap-binding frequency & 0.00\% & 10.40\% \\
Fixed-point $P(\mathrm{bind})$ & 2.92e-58 & 5.32e-02 \\
Fixed-point cap slack & 18.2$\sigma$ & 1.5$\sigma$ \\
Realized mean load & 1666.4 & 16.6 \\
Exact conditional mean load & 1666.2 & 16.5 \\
Fluid mean load & 1666.2 & 17.5 \\
Fluid-minus-exact gap & 0.000 & 0.948 \\
Gap as share of fluid mean & 0.000\% & 5.432\% \\
\midrule
Route cost model & \multicolumn{2}{r}{$g_i(x)=x$, $f_i(x)=c_i x$} \\
Route $\tau$, $b$, max $c_i$ & \multicolumn{2}{r}{1.0, 5.0, 0.5} \\
Greedy accepted transactions & \multicolumn{2}{r}{2000} \\
Route-implied $\Theta^{\max}=q|\mathcal R^\infty|$ & \multicolumn{2}{r}{10000} \\
Greedy min/max utilization & \multicolumn{2}{r}{20.0\% / 100.0\%} \\
\bottomrule
\end{tabular}

%% file: Sections/mpc-tfm.tex
\section{MPC-Assisted Transaction Fee Mechanisms}
\label{app:mpc-assisted-tfm}

This section recalls the MPC-assisted TFM model of \cite{shi2023_tfm} for the finite block-size setting.  The model abstracts an MPC protocol among validators (originally miners in \cite{shi2023_tfm}) as an ideal functionality $\mathcal F_{\rm MPC}$ that implements the TFM faithfully.
\subsection{The MPC-Assisted TFM Game}

Fix a block-size parameter $\Theta$.  In the MPC-assisted model, a TFM is specified by three possibly randomized rules: an allocation rule, a payment rule, and a validator-revenue rule.  The game proceeds as follows.

\begin{enumerate}
    \item Each player may act through zero or more identities.  Each identity submits one bid to $\mathcal F_{\rm MPC}$; let $\boldsymbol b$ be the resulting bid vector.
    \item The ideal functionality $\mathcal F_{\rm MPC}$ runs the allocation rule on $\boldsymbol b$ and selects at most $\Theta$ bids to confirm.
    \item The functionality runs the payment rule and outputs a payment vector $(p_1,...,p_n)$, where unconfirmed bids pay $0$.
    \item The functionality runs the validator-revenue rule and outputs total validator revenue $\mu(\boldsymbol b)$, which is distributed among the validators.  Budget feasibility requires that the total validator revenue is no larger than the total payments collected from confirmed bids.
\end{enumerate}

The important difference from the plain TFM model is that validators do not choose the block contents after observing the bids.  The allocation, payment, and revenue rules are enforced by $\mathcal F_{\rm MPC}$.  Thus the validator-side strategic power that remains in the ideal model is the ability to participate through identities and submit fake bids, not the ability to override the allocation rule.

\subsection{Strategy Space and Utility}

An honest user with value $v_i$ participates through a single identity and submits the bid $v_i$.  A strategic user may instead submit no bid, one non-truthful bid, or multiple bids under multiple identities.  Strategic users may not impersonate honest users' identities.

An honest validator does not submit bids and simply receives its share of the validator revenue.  A strategic validator, or a coalition controlling a $\rho$ fraction of validators, may submit fake bids through arbitrary identities.  In the $\mathcal F_{\rm MPC}$-hybrid model used by \cite{shi2023_tfm}, such a validator coalition cannot choose which bids are allocated once the bid vector is fixed.

If user $i$'s transaction is confirmed and pays $p_i$, its utility is $v_i - p_i$.
If the transaction is not confirmed, its utility is $0$. A validator's utility is its revenue.  The utility of a coalition $\C$ is the sum of the utilities of its members.

The incentive notions used by \cite{shi2023_tfm} are the following.  UIC requires that no individual user can increase its expected utility by deviating from truthful bidding.  MIC requires that no coalition controlling up to a $\rho$ fraction of validators can increase its validator revenue by deviating.  $(\rho,c)$-SCP requires that no coalition consisting of up to a $\rho$ fraction of validators and up to $c$ users can increase its joint utility by deviating.  We state the ex-post form needed for our impossibility results; it follows from the Bayesian impossibilities in \cite{shi2023_tfm} by applying them to point-mass distributions over fixed bid vectors.

\subsection{MPC-Assisted Impossibility Theorems}

We use the following theorem.

\begin{theorem}[MPC-assisted zero validator revenue; Theorem C.5 of \cite{shi2023_tfm}]
\label{thm:mpc-zero-validator-revenue}
For any possibly randomized MPC-assisted TFM, if the mechanism satisfies ex-post UIC, ex-post MIC against a $\rho$-sized validator coalition, and ex-post SCP against a $(\rho,1)$-sized validator-user coalition, then for every bid vector $\boldsymbol b$,
\[
    \mu(\boldsymbol b) = 0.
\]
\end{theorem}

Thus, if the induced MPC-assisted mechanism satisfies the corresponding ex-post UIC, MIC, and $1$-SCP properties on every fixed input, its validator revenue must be zero on every fixed input.

Here non-trivial utility means that, for some input distribution and some participant, the mechanism gives positive expected utility.  Intuitively, once the validator revenue is forced to zero, a two-user coalition can exploit congestion: one colluding user can drop or alter its bid to improve the other colluding user's confirmation probability, and strict two-user SCP rules out any mechanism that creates positive utility under finite block space.